\documentclass[10pt,a4paper]{article}

\usepackage{amsmath,amsgen,latexsym}
\usepackage{amstext,amssymb,amsfonts,latexsym}
\usepackage{theorem}
\usepackage{pifont}
\usepackage[dvips]{graphics,epsfig}
\usepackage{graphicx}

 \newcommand{\bs}{\bigskip}
 \newcommand{\ms}{\medskip}
 \newcommand{\n}{\noindent}
 \newcommand{\s}{\smallskip}
 \newcommand{\hs}[1]{\hspace*{ #1 mm}}
 \newcommand{\vs}[1]{\vspace*{ #1 mm}}

 \newcommand{\cent}{{|}\!\!\mathrm{c}}
 \newcommand{\dollar}{\$}

 \newcommand{\nat}{\mathbb{N}}
 
 \newcommand{\integer}{\mathbb{Z}}

 \newcommand{\complex}{\mathbb{C}}
 \newcommand{\appcomplex}{\tilde{\mathbb{C}}}

\newcommand{\algebraic}{\mathbb{A}}

 \newcommand{\ie}{\textrm{i.e.},\hspace*{2mm}}
 \newcommand{\eg}{\textrm{e.g.},\hspace*{2mm}}
 
 \newcommand{\etalc}{\textrm{et al.}}

 \newcommand{\MM}{{\cal M}}

 \newcommand{\PP}{{\cal P}}
 
 \newcommand{\RR}{{\cal R}}
 
 \newcommand{\VV}{{\cal V}}

 \newcommand{\p}{\mathrm{{P}}}
 \newcommand{\np}{\mathrm{{NP}}}

 \newcommand{\am}{\mathrm{{AM}}}

 \newcommand{\pspace}{\mathrm{{PSPACE}}}

 \newcommand{\qip}{\mathrm{QIP}}
 \newcommand{\ip}{\mathrm{IP}}

 \def\bbox{\vrule height6pt width6pt depth1pt}

\theoremstyle{plain}
\theoremheaderfont{\bfseries}
 \newtheorem{theorem}{Theorem}[section]
 \newtheorem{lemma}[theorem]{Lemma}
 \newtheorem{proposition}[theorem]{Proposition}

 {\theorembodyfont{\rmfamily}
\newtheorem{definition}[theorem]{Definition}}
 {\theorembodyfont{\rmfamily} }
 {\theorembodyfont{\rmfamily} }

 \newtheorem{claim}{Claim}

 \newenvironment{proof}{\par \noindent
            {\bf Proof. \hs{2}}}{\hfill$\Box$ \vspace*{3mm}}

 \newenvironment{proofof}[1]{\vspace*{5mm} \par \noindent
         {\bf Proof of #1.\hs{2}}}{\hfill$\Box$ \vspace*{3mm}}

 \newcommand{\ceilings}[1]{\lceil #1 \rceil}
 \newcommand{\floors}[1]{\lfloor #1 \rfloor}
 \newcommand{\pair}[1]{\langle #1 \rangle}

 \newcommand{\qubit}[1]{| #1 \rangle}
 \newcommand{\bra}[1]{\langle #1 |}
 \newcommand{\ket}[1]{| #1 \rangle}

\newif\ifnotesw\noteswtrue
   {\ifnotesw\marginpar[\hfill\(\top\)]{\(\top\)}\fi}%
      {\ifnotesw\marginpar[\hfill\(\bot\)]{\(\bot\)}\fi}
      
\newcommand{\mnote}[1]%
   {\ifnotesw\marginpar%
	  [{\scriptsize\begin{minipage}[t]{\marginparwidth}
	  \raggedleft#1%
		  \end{minipage}}]%
	  {\scriptsize\begin{minipage}[t]{\marginparwidth}
	  \raggedright#1%
		  \end{minipage}}%
    \fi}

\newcommand{\ignore}[1]{}

\begin{document}

\pagestyle{plain}


\begin{center}
{\Large {\bf Interactive Proofs with Quantum Finite Automata}}
\footnote{This paper expands the second half part
of the extended abstract having appeared in the Proceedings
of the 9th International Conference on Implementation and Application of Automata (CIAA 2004), Lecture Notes in Computer Science, Vol.3317, pp.225--236, Springer-Verlag, Kingston, Canada, July 22--24, 2004.
An extension of the first half part already appeared in Journal of Computer and System Sciences, vol.75, pp.255--269, 2009.
This work was in part supported by the Natural Sciences and Engineering Research Council of Canada.} \bs\ms\\

{\sc Harumichi Nishimura\footnote{Present Affiliation:
Department of Computer Science and Mathematical Informatics,
Graduate School of Information Science, Nagoya University,
Chikusa-ku, Nagoya, Aichi, 464-8601 Japan.}
and Tomoyuki Yamakami\footnote{Present Affiliation: Department of Information Science, University of Fukui, 3-9-1 Bunkyo, Fukui, 910-8507 Japan.}} \bs\\
\end{center}
\bs

\begin{center}
{\bf Abstract}
\end{center}
\begin{quote}
Following an early work of Dwork and Stockmeyer on interactive proof systems whose verifiers are two-way probabilistic finite automata, the authors initiated in 2004 a study on the computational power of quantum interactive proof systems whose verifiers are particularly limited to quantum finite automata. As a follow-up to the authors' early journal publication [J. Comput. System Sci., vol.75, pp.255--269, 2009], we further investigate the quantum nature of interactions between provers and verifiers by studying how various restrictions on quantum interactive proof systems affect the language recognition power of the proof systems. In particular, we examine three intriguing restrictions that (i) provers always behave in a classical fashion, (ii) verifiers always reveal to provers the information on next moves, and (iii) the number of interactions between provers and verifiers is bounded.

\s

\n{\sf Keywords:} finite automaton, interactive proof system, quantum computing, classical prover, quantum prover, interaction
\end{quote}

\section{Overview}\label{sec:QFA}

Quantum mechanics has provided an unconventional means to fast computing and secure communication since early 1980s.
The potential power of quantum interactions between two parties---a prover and a verifier---motivated the authors \cite{NY04} in 2004
to investigate a quantum analogy of an early work of Dwork and Stockmeyer \cite{DS92}
on classical {\em interactive proof (IP) systems} whose verifiers are restricted to two-way probabilistic finite automata (or 2pfa's, in short).
Such weak verifiers can represent
computations that utilize only a finite amount of memory space.
These IP systems are a special case of a much wider scope of space-bounded IP systems of Condon \cite{Con93}. It has been shown that those IP systems behave quite differently from time-bounded IP systems.
Although verifier's power is limited to 2pfa's, as Dwork and Stockmeyer dexterously demonstrated, 2pfa-verifier IP systems turn out to be significantly powerful, because a number of interactions with a mighty prover can enhance the 2pfa-verifier's ability to recognize much more complicated languages than the 2pfa's alone recognize. To describe succinctly such restricted IP systems, Dwork and Stockmeyer introduced a special notation $\ip(\pair{restriction})$ for a class of languages that are recognized by bounded-error IP systems with weak verifiers under restrictions specified by  $\pair{restriction}$. For instance, the notation $\ip(2pfa,poly\mbox{-}time)$ expresses the language class characterized by IP systems with 2pfa verifiers operated in {\em expected} polynomial time. The seminal work of Dwork and Stockmeyer further
studied numerous subjects (private coins versus public coins, etc.) for
a basis model of 2pfa-verifier IP systems.
To express Babai's Arthur-Merlin proof systems \cite{Bab85}, in particular, they introduced another notation $\am(\pair{restriction})$ under restrictions given in $\pair{restrictions}$.
An immediate advantage of studying such weak verifier models is to prove certain types of separations and collapses among the associated complexity classes {\em without any unproven assumption}. For instance, Dwork and Stockmeyer successfully separated $\ip(2pfa,poly\mbox{-}time)$ from $\am(2pfa)$, and $\am(2pfa,poly\mbox{-}time)$ from $\am(2pfa)$. Another significant advantage is that their restricted models make it easier to analyze the behaviors of two parties---a prover and a verifier---during a usually complicated communication process between them.

Unlike Watrous' model \cite{Wat03} of circuit-based {\em quantum interactive proof (QIP) systems},
a basic QIP model of the authors \cite{NY04,NY09} uses a {\em measure-many two-way quantum finite automaton} (or 2qfa) of Kondacs and Watrous \cite{KW97} as a weak verifier who can communicate with a mighty {\em quantum prover} through a quantum communication bulletin board (implemented as a {\em communication cell} that holds at each moment one symbol from a communication alphabet). They also considered two additional variants of verifiers: a {\em measure-once one-way quantum finite automaton} (or mo-1qfa) of Moore and Crutchfield \cite{MC00} and a {\em measure-many one-way quantum finite automaton}  (or 1qfa) of Kondacs and Watrous \cite{KW97}. An initial study of QIP systems with those weak verifiers reveals their noticeable strength.
Let us recall from \cite{NY09} a general notation   $\qip(\pair{restrictions})$ for bounded-error QIP systems under restrictions indicated by $\pair{restrictions}$, analogous to the aforementioned notation $\ip(\pair{restriction})$ of Dwork and Stockmeyer. For instance, $\qip(2qfa)$ is obtained by restricting all verifiers to 2qfa's. Likewise, a use of mo-1qfa verifiers and 1qfa verifiers introduces the language classes $\qip(mo\mbox{-}1qfa)$ and $\qip(1qfa)$, respectively.

The power of quantum interaction was exemplified in  \cite{NY09}.
(i) With mo-1qfa verifiers, it holds that   $\mathrm{MO\mbox{-}1QFA}\subsetneqq
\qip(mo\mbox{-}1qfa) \subsetneqq \mathrm{REG}$, where $\mathrm{REG}$ is the class of regular languages and $\mathrm{MO\mbox{-}1QFA}$ is the class of all languages recognized by bounded-error mo-1qfa's. (ii) With 1qfa verifiers, we obtain  $\mathrm{1QFA}\subsetneqq \qip(1qfa) = \mathrm{REG}$, where $\mathrm{1QFA}$ is the class of all languages recognized by bounded-error 1qfa's.
(iii) With 2qfa verifiers, it holds that $\qip(1qfa) \subsetneqq \qip(2qfa,poly\mbox{-}time)\nsubseteq \mathrm{AM}(2pfa)$.
In addition, they showed that $\mathrm{2QFA}_{\algebraic}\subseteq \p$ and $\qip_{\appcomplex}(2qfa,poly\mbox{-}time)\subseteq\np$, where  $\algebraic$ and $\appcomplex$ respectively indicate that all (transition) amplitudes of qfa verifiers are algebraic complex numbers and polynomial-time {\lq\lq}approximable{\rq\rq} complex numbers. The last result clearly contrasts with the following classical containments: $\am(2pfa)\subseteq \ip(2pfa,poly\mbox{-}time)\subseteq\pspace$ \cite{DS92}.


We intend to continue a study of qfa-verifier QIP systems from various aspects of computational complexity. In particular,
This paper aims at presenting three intriguing subjects that were briefly discussed
in \cite{NY04} but have been excluded from the authors' early journal publication \cite{NY09}.
In this current publication, we shall examine strengths and weaknesses of qfa-verifier QIP systems by observing how various restrictions on the QIP systems affect their power of language recognition. Our investigation
is focused on the following three selected subjects.

\ms
\n{\bf 1. Classical provers versus quantum provers.}\hs{2}
In the model of QIP systems in \cite{NY09},
provers are basically quantum machines, which
can apply any predetermined unitary operators. In contrast, provers in IP systems of Dwork and Stockmeyer \cite{DS92} are essentially {\lq\lq}probabilistic{\rq\rq} machines, which can flip privately owned
coins and decide what messages to send back to 2pfa verifiers. These probabilistic machines, however, are known to be reduced to deterministic machines, which are naturally associated with unitary operators whose entries are only $0$s and $1$s because the provers can use an unlimited amount of private memory storage. For convenience, we briefly call such provers {\em classical provers}; in contrast, we call standard provers {\em quantum provers}.
Naturally, we raise a simple question of whether our quantum
provers are truly different in recognition power from the aforementioned
classical provers.

It appears that a classical prover helps a 2qfa verifier much more than a quantum prover does.
For instance, the language $Center = \{ x1y \mid x,y\in\{0,1\}^*, |x|=|y|\}$ is not yet known to belong to $\qip(2qfa)$; however, interactions with a classical prover allow a 2qfa verifier to recognize this particular
language.
Such a strength of using classical provers stems from a simple fact that an analysis of classical-prover QIP protocols is much easier than that of  quantum-prover ones. This paper further shows the following containments and separations concerning classical provers, where we use the restriction $\pair{c\mbox{-}prover}$ to indicate the use of classical provers.
(i) $\qip(1qfa) \subseteq \qip(1qfa,c\mbox{-}prover)$. (ii)
$\am(2pfa)\subsetneqq \qip(2qfa,c\mbox{-}prover)$. (iii)
$\am(2pfa,poly\mbox{-}time)\subsetneqq \qip(2qfa,poly\mbox{-}time,c\mbox{-}prover)\nsubseteq \am(2pfa)$.
A core argument used for these results is a technical construction of appropriate QIP protocols that recognize target languages. All the above results will be  presented in Section \ref{sec:classical-prover}.

\ms
\n{\bf 2. Public information versus private information.}\hs{2}
As noted earlier, Dwork and Stockmeyer \cite{DS92} examined two different types of IP systems---$\ip(\pair{restriction})$ and $\am(\pair{restriction})$---when verifiers are limited to 2pfa's.
The latter IP system $\am(\pair{restriction})$ refers to Babai's model of Arthur-Merlin proof systems \cite{Bab85} (also known as {\lq\lq}public-coin{\rq\rq} IP systems), in which verifiers
flip fair coins to decide next moves and reveal their outcomes, essentially showing a state of next internal configurations to mighty provers, who can keep track of the verifier's past configurations.
In comparison, standard IP systems are often referred to as {\lq\lq}private-coin{\rq\rq} IP systems.
A question concerning {\lq\lq}public coins{\rq\rq}  versus {\lq\lq}private  coins,{\rq\rq} which was a key subject in \cite{DS92},
is essentially whether the prover obtains {\lq\lq}complete information{\rq\rq}  or {\lq\lq}partial information{\rq\rq} on the configurations of the verifier.
It was shown in \cite{DS92} that $\mathrm{2PFA}\neq \am(2pfa)\neq \ip(2pfa)$, highlighting a clear difference between public information and private information.

Likewise, we shall introduce a similar {\lq\lq}publicness{\rq\rq} notion to qfa-verifier QIP systems by demanding verifiers
to reveal their next moves to provers at every step.
To express the public-coin analogue of QIP systems, we use the notation $\pair{public}$ as a restriction to such systems.\footnote{These new QIP systems might be possibly expressed as $\mathrm{QAM}(\pair{restriction})$ analogous to $\am(\pair{restriction})$.}
With this notation, for instance, $\qip(1qfa,public)$ denotes
the language class obtained from $\qip(1qfa)$
by restricting verifiers to publicly announcing their next moves.
It turns out that public QIP systems remain significantly powerful.
To be more precise, we shall prove the following three class relations.
(i) $\mathrm{1RFA}\subsetneqq \qip(1qfa,public) \nsubseteq \mathrm{1QFA}$,  (ii) $\qip(2qfa,public,poly\mbox{-}time) \nsubseteq \am(2pfa,poly\mbox{-}time)$,  and (iii) $\qip(2qfa,public,c\mbox{-}prover) \nsubseteq \am(2pfa,poly\mbox{-}time)$, where $\mathrm{1RFA}$ is the language family induced by {\em one-way (deterministic) reversible finite automata} (1rfa's, in short) \cite{AF98}. In Section \ref{sec:public-coin}, we shall discuss those results in details.

\ms
\n{\bf 3. The number of interactions between a prover and a verifier.}\hs{2}
As suggested in \cite[Section 6]{NY09}, the number of interactions between a prover and a verifier in a weak-verifier QIP system may serve as a complexity measure of classifying various languages.
Unlike Dwork-Stockmeyer IP systems, the original QIP systems of the authors \cite{NY09} were introduced as to force two parties---a prover and a verifier---to communicate with each other {\em at every step} and,
through Sections \ref{sec:classical-prover}--\ref{sec:public-coin}, we shall take this definition of QIP systems.
To study the precise effect of interactions, nevertheless, we need to modify this original model slightly so that the verifier can interact with the prover only when he needs any help from the prover.
To express those new QIP systems and their corresponding language classes,
we invent two new notations $\qip^{\#}(\pair{restrictions})$ and $\qip^{\#}_{k}(\pair{restrictions})$, where $k$ indicates the maximal number of iterations made during each computation.
In Section \ref{sec:interaction}, we shall prove that
$\qip^{\#}_0(1qfa) \subsetneqq\qip^{\#}_1(1qfa)
\subsetneqq\qip^{\#}(1qfa)$. The first separation between $\qip^{\#}_0(1qfa)$ and $\qip^{\#}_1(1qfa)$ comes from the fact that the language $Odd =\{0^m1z\mid z\in\{0,1\}^*,\,\text{$z$ has an odd number of $0$s}\,\}$
belongs to $\qip^{\#}_1(1qfa)$ but it is not in $\qip^{\#}_0(1qfa)$ since $\qip^{\#}_{0}(1qfa)$ coincides with $\mathrm{1QFA}$.
In contrast, the second separation of $\qip^{\#}(1qfa)$ from $\qip^{\#}_{1}(1qfa)$ is exemplified by the language $Zero= \{x0\mid x\in\{0,1\}^*\}$; however, the proof of $Zero\not\in \qip^{\#}_{1}(1qfa)$ is much more involved than the proof of $Zero\not\in\mathrm{1QFA}$ that appears in \cite{KW97}.

\section{QFA-Verifier QIP Systems}\label{application}

Throughout this paper, $\complex$ denotes the set of all
{\em complex numbers} and $\imath$ is $\sqrt{-1}$.
Let $\nat$ be the set of all {\em natural numbers}
(\ie nonnegative integers) and set $\nat^{+}= \nat-\{0\}$.
Given two integers $m$ and $n$ with $m\leq n$, the {\em integer interval}  $[m,n]_{\integer}$ is the set
$\{m,m+1,m+2,\ldots,n\}$ and $\integer_{n}$ in particular denotes the set $[0,n-1]_{\integer}$.
All {\em logarithms} are to base 2 and all {\em polynomials} have integer coefficients.
An {\em alphabet} is a finite nonempty set of {\lq\lq}symbols{\rq\rq} and our input alphabet $\Sigma$ is not necessarily limited to $\{0,1\}$ throughout this paper.
Following the standard convention, $\Sigma^*$ denotes the set of all {\em finite} sequences of symbols from $\Sigma$,
and we write $\Sigma^{n}=\{x\in\Sigma^*\mid |x|=n\}$, where $|x|$ denotes the length of $x$.
Opposed to the notation $\Sigma^*$, $\Sigma^{\infty}$ stands for the set of all {\em infinite} sequences,
each of which consists of symbols from $\Sigma$. For any symbol $a$ in $\Sigma$,
$a^{\infty}$ denotes an element of $\Sigma^{\infty}$, which is the infinite sequence made only of $a$.
We assume the reader's familiarity with classical automata theory and the basic concepts of quantum
computation (refer to, \eg \cite{Gru99,HMU01,NC00} for its foundation).
As underlying computation device, we extensively use {\em measure-many one-way quantum finite automata} (or 1qfa's, in short) and of {\em measure-many two-way quantum finite automata} (or 2qfa's), where we assume the reader's familiarity with the definitions of those quantum automata \cite{KW97}.
The reader may refer to \cite{DS92} for the formal definitions of the classes $\ip(\pair{restriction})$ and $\am(\pair{restriction})$ described in Section \ref{sec:QFA}.

\subsection{Basic Model of QIP Systems}\label{sec:basic-def}

Let us review the fundamental definition of QIP system of the authors  \cite{NY04,NY09}, in which verifiers are particularly limited to
quantum finite automata.
In Sections \ref{sec:public-coin}--\ref{sec:interaction}, we shall further
restrict the behaviors of those verifiers
as well as provers in order to obtain three major variations of our basic QIP systems.
The reader may refer to \cite[Section 3.2]{NY09} for a brief discussion on the main difference between QIP systems
based on uniform quantum circuits and those based on quantum finite automata.

We use the notation $(P,V)$ to denote a {\em QIP protocol} executed by prover $P$ and verifier $V$ (whose schematic diagram is illustrated in Figure \ref{fig:scheme}).
Unless there is any confusion, we also use the same notation $(P,V)$ to mean a {\em QIP system}  with the prover $P$ and the verifier $V$.
The 2qfa verifier $V = (Q,\Sigma\cup\{\cent,\dollar\},\Gamma,\delta,q_0,Q_{acc},Q_{rej})$ is a 2qfa  specified by a finite set $Q$ of verifier's inner states, an input alphabet $\Sigma$ and a
verifier's transition function $\delta$, equipped further with a shared communication cell using a communication alphabet $\Gamma$. The set $Q$ is the union of
three mutually disjoint subsets $Q_{non}$, $Q_{acc}$, and $Q_{rej}$,
where any states in $Q_{non}$, $Q_{acc}$, and $Q_{rej}$ are respectively
called a {\em non-halting inner state}, an {\em accepting inner state},
and a {\em rejecting inner state}.
In contrast to $Q_{non}$,  inner states in $Q_{acc}\cup Q_{rej}$ are simply called {\em halting inner states}.
In particular, $Q_{non}$ contains a so-called {\em initial inner state}
$q_0$.
An input tape is indexed by natural numbers (where the first cell is indexed $0$). Two designated symbols
$\cent$ and $\$$ not appearing in $\Sigma$, which are called respectively the
{\em left endmarker} and the {\em right endmarker},
mark the left end and the right end of the input on the input tape.
For our convenience, set
 $\check{\Sigma}=\Sigma\cup\{\cent,\$\}$. Assume also that $\Gamma$
contains a blank symbol $\#$ with which the system $(P,V)$ begins in the communication cell.

The {\em verifier's transition
function} $\delta$ is a map from $Q\times\check{\Sigma}\times\Gamma\times
Q\times\Gamma\times\{0,\pm 1\}$ to $\complex$ and is interpreted as
follows. For any $q,q'\in Q$, $\sigma\in\check{\Sigma}$,
$\gamma,\gamma'\in\Gamma$,
and $d\in\{0,\pm 1\}$, the complex number
$\delta(q,\sigma,\gamma,q',\gamma',d)$
specifies the transition amplitude with which the verifier $V$  in state $q$
scanning symbol $\sigma$ on the input tape and symbol $\gamma$ in the communication
cell changes $q$ to $q'$, replaces $\gamma$ with $\gamma'$,
and moves his tape head on the input tape in direction $d$. When the tape head is located in a cell indexed $t$, it must move to the cell indexed $t+d$.

At the beginning of the computation,
an input string $x$ over $\Sigma$ of length $n$ is written orderly from
the first cell to the $n$th cell of the input tape. The tape head initially
scans $\cent$ in the $0$th cell. The communication cell holds only a symbol in $\Gamma$ and initially $\#$ is written in the cell.
As in the original definition of 2qfa in \cite{KW97}, our input tape
is {\em circular}; that is, whenever the verifier's tape head scanning $\cent$
($\$$, resp.) on the input tape moves to the left (right, resp.), the tape
head reaches to the right end (resp. left end) of the input tape.

\begin{figure}[t]
\begin{center}
\centerline{\psfig{figure=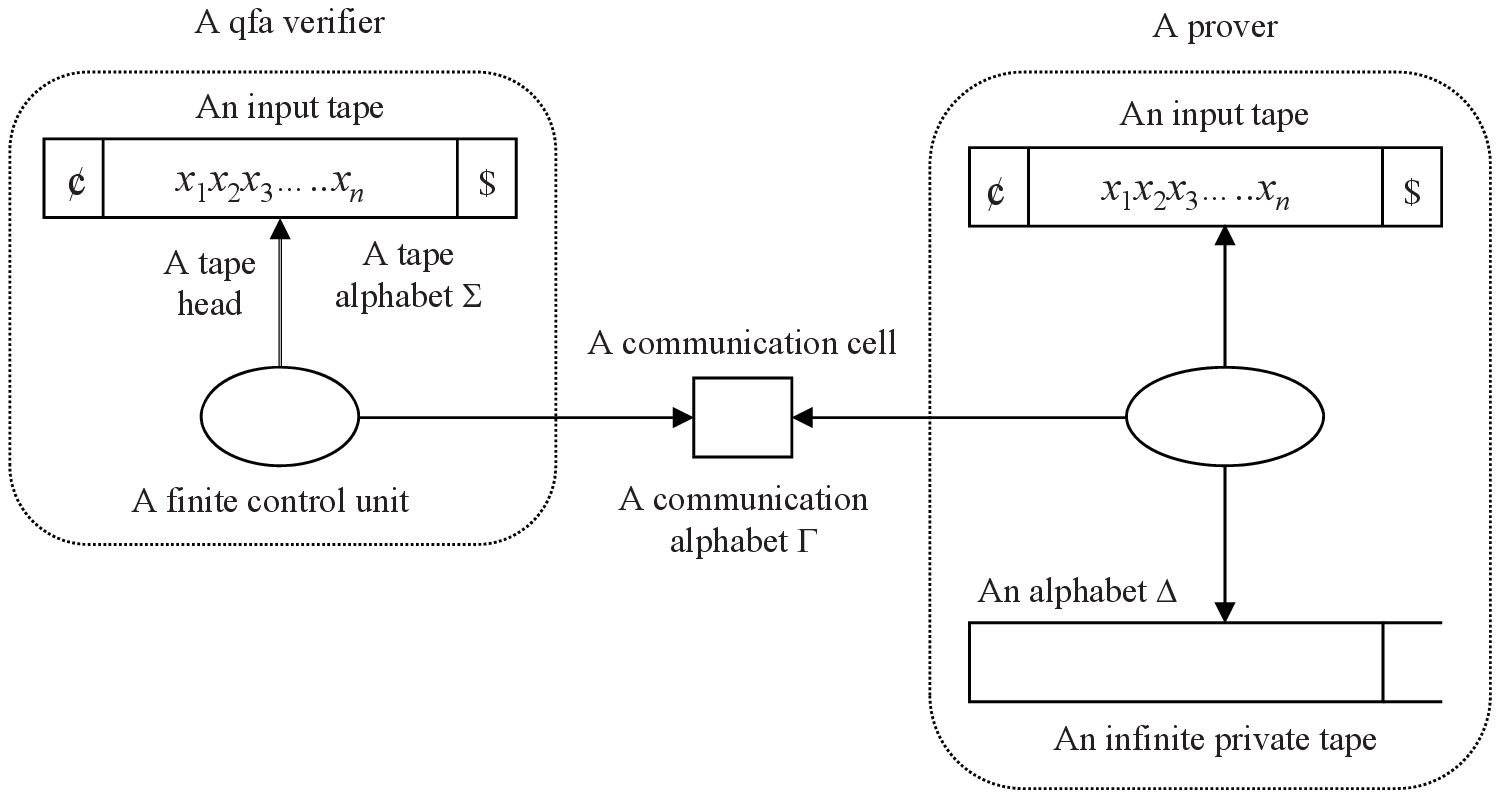,height=5.0cm}}
\caption{{\small A schematic view of a QIP system with a qfa verifier}}\label{fig:scheme}
\end{center}
\end{figure}

Next, we shall explain two concepts of (global) configuration and visible configuration.  A {\em (global) configuration} of the QIP protocol $(P,V)$
is a snapshot of a {\lq\lq}computation{\rq\rq} of the protocol,
comprising the following visible configurations
of the two players. Each player can see only his portion of a global
configuration.
A {\em visible configuration} of the verifier $V$ on an input of
length $n$ is
represented by a triplet
$(q,k,\gamma)\in Q\times\integer_{n+2}\times\Gamma$,
which indicates that the verifier is in state $q$, the content of the
communication cell is $\gamma$, and the verifier's tape head position is
$k$ on the input tape. Let $\VV_n$ and $\MM$ be respectively the Hilbert
spaces spanned by the computational bases
$\{\qubit{q,k}\mid (q,k)\in Q\times\integer_{n+2}\}$ and
$\{\qubit{\gamma}\mid \gamma\in\Gamma\}$. The Hilbert space
$\VV_n\otimes\MM$ is called the {\em verifier's visible configuration
space} on inputs of length $n$.

When a direction of every tape head's move is completely determined by the pair $(q',\gamma')$ of state $q$ and content $\gamma'$ of the communication cell after applying the associated transition,
we call the verifier {\em unidirectional} and use the following notation to simplify a description of the verifier's strategy (or transition function) $\delta$.
Assuming that the verifier $V$ is unidirectional, a {\em $(Q\times\Gamma)$-transition} $V_{\sigma}$ of
$V$ on input symbol $\sigma\in\check{\Sigma}$ is a unitary operator acting on $\mathrm{span}\{\qubit{q}\qubit{\gamma}\mid q\in Q,\gamma\in \Gamma\}$ of the form $V_{\sigma}\qubit{q,\gamma} = \sum_{q',\gamma'}\alpha_{q,\gamma,q',\gamma'}\qubit{q',\gamma'}$,
where $q'\in Q$, $\gamma'\in\Gamma$, and $\alpha_{q,\gamma,q',\gamma'}\in\complex$.
This means that, if $M$ scans $\sigma$ in state $q$ with $\gamma$, it changes $q$ to $q'$, $\gamma$ to $\gamma'$ and the tape head moves in direction $d$. Since $d$ is determined uniquely from $(q',\gamma')$,
we often express the head direction as $D(q',\gamma')=d$. When $d$ is independent of $\gamma$, by contrast,
we succinctly write $D(q')=d$. The verifier's strategy $\delta$ can be expressed as $\delta(q,\sigma,\gamma,q',\gamma',d) = \bra{q',\gamma'}V_{\sigma}\ket{q,\gamma}$ if $D(q',\gamma')=d$, and $0$ otherwise.
Notice that most 2qfa's constructed in the past literature (for instance, \cite{KW97}) actually satisfy this unidirectional condition.

For any input $x$ in $\Sigma^*$ of length $n$,
$\delta$ automatically induces the linear
operator $U_\delta^x$ acting on the Hilbert space
$\VV_n\otimes\MM$ defined by
$U_\delta^x\qubit{q,k,\gamma}= \sum_{q',\gamma',d}
\delta(q,x_{(k)},\gamma,q',\gamma',d)\qubit{q',k',\gamma'}$, where $x_{(0)}=\cent$, $x_{(n+1)}=\dollar$,
$x_{(i)}$ ($i\in[1,n]_{\integer}$) is the $i$th symbol in $x$,
and $k'=k+d\ (\mbox{mod }n+2)$.
The verifier is called {\em well-formed} if $U_\delta^x$ is
unitary on $\VV_n\otimes\MM$ for every string $x\in\Sigma^*$. Since
we are interested only in well-formed verifiers, we henceforth assume that all verifiers should be well-formed.

Let $x$ be any input $x$ of length $n$. The 2qfa verifier $V$ starts with
the initial quantum state $\qubit{q_0,0,\#}$.
A single step of the verifier on  $x$ consists of the
following process.
First, $V$ applies his operation $U_\delta^x$ to an existing superposition $\qubit{\phi}$ in $\VV_n\otimes\MM$
and then $U_\delta^x\qubit{\phi}$ becomes a new superposition
$\qubit{\phi'}$. Second, we define  $W_{acc}=\mathrm{span}\{\qubit{q,k,\gamma}\mid
(q,k,\gamma)\in Q_{acc}\times\integer_{n+2}\times\Gamma\}$,
$W_{rej}=\mathrm{span}\{\qubit{q,k,\gamma}\mid (q,k,\gamma)\in Q_{rej}\times\integer_{n+2}\times\Gamma\}$,
and $W_{non}=\mathrm{span}\{\qubit{q,k,\gamma}\mid (q,k,\gamma)\in Q_{non}\times\integer_{n+2}\times\Gamma\}$.
Moreover, let $k_{acc}$, $k_{rej}$, and $k_{non}$ be respectively the positive
numbers representing {\lq\lq}accepting,{\rq\rq}
{\lq\lq}rejecting,{\rq\rq} and {\lq\lq}non halting.{\rq\rq}
The new superposition $\qubit{\phi'}$ is then measured by
the {\em observable}
$k_{acc}E_{acc}+ k_{rej}E_{rej}+ k_{non}E_{non}$, where $E_{acc}$,
$E_{rej}$, and $E_{non}$ are respectively the projection operators
onto $W_{acc}$, $W_{rej}$, and $W_{non}$.
Provided that $\qubit{\phi'}$ is expressed as
$\qubit{\psi_1}+\qubit{\psi_2}+\qubit{\psi_3}$ for certain three vectors
$\qubit{\psi_1}\in W_{acc}$, $\qubit{\psi_2}\in W_{rej}$, and
$\qubit{\psi_3}\in W_{non}$, we say that, at this step, {\em $V$
accepts $x$ with probability} $\|\qubit{\psi_1}\|^2$ and {\em
rejects $x$ with probability} $\|\qubit{\psi_2}\|^2$. Only the
non-halting superposition $\qubit{\psi_3}$
continues to the next step and $V$ is said to {\em continue (to
the next step) with probability} $\|\qubit{\psi_3}\|^2$. The
probability that $x$ is accepted (rejected, resp.) within the first
$t$ steps is thus the sum, over all $i\in[1,t]_{\integer}$, of
the probabilities with which $V$ accepts (rejects, resp.) $x$ at
the $i$th step.

When the verifier is a 1qfa, the verifier's transition function $\delta$
must satisfy the following additional condition: for
all $q,q'\in Q$, $\sigma\in\check{\Sigma}$, and $\gamma,\gamma'\in\Gamma$,
it must hold that $\delta(q,\sigma,\gamma,q',\gamma',d) =0$ if $d\neq +1$ (\ie the tape
head does not move to the right). Unlike 2qfa verifiers,
a 1qfa verifier must stop running after applying $\delta$ at scanning $\dollar$ and then performing a projective measurement. In other words, the 1qfa verifier completely stops by the time the verifier's tape head moves off $\$$
(thus, the tape head actually stops at $\cent$ since the input tape is circular).
Therefore, on any input $x$, the 1qfa verifier halts in at most $|x|+2$ steps.

In contrast to the verifier, the prover $P$ has a semi-infinite private
tape and accesses the input $x$ and the communication cell. For the prover's private tape, let $\Delta$
be a tape alphabet, which includes a special blank symbol $\#$.
The prover is assumed to alter only a {\lq\lq}finite{\rq\rq} initial
segment of his private tape at every step. Let $\PP$ be the
Hilbert space spanned by $\{\qubit{y}\mid y\in \Delta^{\infty}_{fin}\}$,
where $\Delta^{\infty}_{fin}$ is the set of all infinite sequences
of tape symbols containing only a finite number of non-blank symbols.
The {\em prover's visible configuration space} is the Hilbert space $\MM\otimes\PP$.
Formally, the prover $P$ is specified by a series
$\{U_{P,i}^x\}_{x\in\Sigma^*,i\in\nat^+}$ of unitary operators, each of which acts on the  prover's visible
configuration space, such that $U_{P,i}^x$ is of the form
$S_{P,i}^x\otimes I$, where $\dim(S_{P,i}^x)$ is finite and $I$ is the
identity operator. Such a series of operators is often called the
{\em prover's strategy} on the input $x$. To refer to the
strategy on $x$, we often use the notation $P_x$; namely,  $P_x=\{U_{P,i}^{x}\}_{i\in\nat^{+}}$. With this notation, the prover can be expressed as $\{P_{x}\}_{x\in\Sigma^*}$.
If the prover has string $y\in\Delta_{fin}^{\infty}$ on his private tape
and scans symbol $\gamma$ in the communication cell, then he applies
$U_{P,i}^x$ to the quantum state $\qubit{\gamma}\qubit{y}$ at the
$i$th step of the prover's turn.
If $U_{P,i}^x\qubit{\gamma}\qubit{y} =\sum_{\gamma',y'}\alpha_{\gamma',y'}^{i}\qubit{\gamma'}\qubit{y'}$,
then the prover changes $y$ into $y'$ and replaces $\gamma$ by $\gamma'$ with amplitude $\alpha_{\gamma',y'}^{i}$.

A {\em (global) configuration} consists of the four items:
$V$'s inner state, $V$'s tape head position, the content of the communication cell, and the content of $P$'s
private tape. We express a superposition of such configurations
of $(P,V)$ on input $x$ as a vector in the Hilbert space $\VV_{|x|}\otimes\MM\otimes\PP$,
which is called the {\em (global) configuration space} of $(P,V)$
on the input $x$; in other words, a (global) configuration is a pure quantum state of the form $\qubit{q,k}\qubit{\gamma}\qubit{y}$, indicating that $V$ is in inner state $q$, its tape head is at cell $k$, $\gamma$ is in the communication cell,
and $P$'s private tape contains $y$.  A global configuration $\xi$ is called a {\em halting configuration} (a {\em non-halting configuration}, resp.) if  $\xi$ contains a halting (non-halting, resp.) inner state of $V$.

A {\em computation} of the QIP protocol $(P,V)$ on the input $x$ constitutes
a series of superpositions of configurations resulting from alternate applications
of unitary operations of the prover $P$ and the verifier $V$ including his projective   measurements  in the following manner.
The computation of $(P,V)$ on  $x$ starts with the global initial configuration $\qubit{q_0,0}\qubit{\#}\qubit{\#^{\infty}}$,
where the verifier is in his initial configuration and the prover's
private tape consists only of the blank symbol $\#$.
The two players $P$ and $V$ apply their unitary
operators  $P_x$ and $U_\delta^x$ (as well as measurements) {\em in
turn}, starting with the verifier's move. A projective measurement is made after every move of the verifier to determine whether $V$ is in a halting inner state.
Through the communication cell, the
two players exchange communication symbols, which naturally cause the two players to be entangled.
When the prover (verifier, resp.) writes symbol $\sigma\in\Gamma$ in the communication cell, we customarily say that the prover (verifier, resp.) {\em sends} $\sigma$ to the verifier (prover., resp.),  where $\sigma$ is viewed as a {\lq\lq}message.{\rq\rq}
More precisely, when $\qubit{q,k}\qubit{\gamma}\qubit{y}$ is a current
global configuration, $V$ changes it into $(U_{\delta}^{x}\otimes I_1)\qubit{q,k}\qubit{\gamma}\qubit{y} = U_{\delta}^{x}\qubit{q,k}\qubit{\gamma} \otimes \qubit{y}$, where $I_1$ is the  identity operator acting on $\PP$. After $V$ applies the  projective  measurement $E_{non}$, the global configuration becomes $(E_{non}\otimes I_1)(U_{\delta}^{x}\otimes I_1)\qubit{q,k}\qubit{\gamma}\qubit{y}$.
Finally, $P$ changes $\qubit{q,k}\qubit{\gamma}\qubit{y}$ into $(I_2\otimes U_{P,i}^{x})\qubit{q,k}\qubit{\gamma}\qubit{y} = \qubit{q,k} \otimes U_{P,i}^{x}\qubit{\gamma}\qubit{y}$, where $I_2$ is the identity operator acting on $\VV_{|x|}$.
A superposition $\qubit{\Phi_i}$ of global configurations at the $i$th step is defined recursively as $\qubit{\Phi_0}=\qubit{q_0,0}\qubit{\#}\qubit{\#^{\infty}}$,  $\qubit{\Phi_{2i+1}} = (E_{non}\otimes I_1) (U_{\delta}^{x}\otimes I_1)\qubit{\Phi_{2i}}$, and $\qubit{\Phi_{2i+2}} = (I_2\otimes U_{P,{i+1}}^{x}) \qubit{\Phi_{2i+1}}$ for every $i\in\nat$. For example,
the superposition of global configurations after the $2i+1$st step becomes
\[
(E_{non}\otimes I_1) (U_{\delta}^{x}\otimes I_1) (I_2\otimes U_{P,i}^{x}) \cdots
(U_{\delta}^{x}\otimes I_1)(I_2\otimes U_{P,1}^{x})
(E_{non}\otimes I_1)(U_{\delta}^{x}\otimes I_1)\qubit{q_0,0}\qubit{\#}\qubit{\#^{\infty}}.
\]
The series $\ket{\Phi_0},\ket{\Phi_1},\ldots$ therefore becomes a computation of $(P,V)$ on $x$.

Given any global configuration $\xi$, a {\em local computation path} ending with (or leading to) $\xi$ in computation $(\qubit{\Phi_0},\qubit{\Phi_1},\ldots,\qubit{\Phi_n})$
of the QIP protocol $(P,V)$ on a given input is a series $(\xi_0,\xi_1,\ldots,\xi_m)$ of global configurations satisfying the following four conditions: $\qubit{\xi_0}=\qubit{q_0,0}\qubit{\#}\qubit{\#^{\infty}}$, $(E_{non}\otimes I_1) (U_{\delta}^{x}\otimes I_1)\qubit{\xi_{2i}}$ contains $\qubit{\xi_{2i+1}}$ with non-zero amplitude for all $i\in[0,\floors{(m-1)/2}]_{\integer}$, $(I_2\otimes U_{P,{i+1}}^{x}) \qubit{\xi_{2i+1}}$ contains $\qubit{\xi_{2i+2}}$ with non-zero amplitude for all $i\in[0,\floors{(m-2)/2}]_{\integer}$, and $\xi_m$ equals $\xi$.
Moreover, a {\em (global) computation path} ending with $\xi$ is a local computation path $(\xi_0,\xi_1,\ldots,\xi_m)$ ending with $\xi$ for which $\qubit{\xi_i}$ appears in $\qubit{\Phi_i}$ with non-zero amplitude for every $i\in[0,m]_{\integer}$.
Each (global) computation path ends when $V$ enters a certain halting inner state along this computation path.
Furthermore, we define
the overall probability that $(P,V)$ {\em accepts} ({\em rejects}, resp.) the input $x$ to be the
limit, as $t\rightarrow\infty$, of the probability that $V$
accepts (rejects, resp.) $x$ within the first $t$ steps. We use the
notation $p_{acc}(x,P,V)$ ($p_{rej}(x,P,V)$, resp.) to denote
the overall acceptance (rejection, resp.) probability of $x$
by $(P,V)$. We say that $V$ {\em always halts with probability $1$} if, for every input
$x$ and every prover $P^{*}$, the QIP protocol $(P^{*},V)$ reaches halting inner states with
probability $1$. In general, $V$ may not always halt with probability $1$.
Notice that, when we discuss the entire {\em running time} of the QIP
system, we count the number of all steps taken by
the verifier (including measurements) {\em as well as the prover}.

Let $a$ and $b$ be any two real numbers in the unit real
interval $[0,1]$ and let $L$ be any language. We say that {\em $L$
has an $(a,b)$-QIP system} $(P,V)$ (or {\em an $(a,b)$-QIP system
$(P,V)$ recognizes} $L$) if $(P,V)$ is a QIP system and the following two conditions hold
for the corresponding QIP protocol $(P,V)$:
\begin{enumerate}\vs{-1}
  \setlength{\topsep}{0mm}%
  \setlength{\itemsep}{1mm}
  \setlength{\parskip}{0cm}%
  
\item {\sf (completeness)} for any $x\in L$, the QIP protocol $(P,V)$ accepts $x$ with probability at least $a$, and

\item {\sf (soundness\footnote{As Lipton \cite{Lip89} demonstrated, this form of the soundness condition cannot be, in general, replaced by the following weaker form: {\lq\lq}$(P,V)$ accepts $x$ with probability at most $1-b$.{\rq\rq} See \cite{DS92} for a discussion.})} for any $x\not\in L$ and any prover $P^*$, the QIP protocol
$(P^*,V)$ rejects $x$ with probability at least $b$.
\end{enumerate}

Note that an $(a,a)$-QIP system has the error probability at most
$1-a$. This paper discusses only the QIP systems whose error
probabilities are bounded from above by certain constants lying in
the real interval $[0,1/2)$. For this reason, we simply say that {\em $L$ has a QIP system} if there exists a constant (an error bound) $\epsilon\in[0,1/2)$ such that $L$ has a $(1-\epsilon,1-\epsilon)$-QIP system.

Given any pair $a,b\in[0,1]$, the notation $\qip_{a,b}(\pair{\RR})$, where $\pair{\RR}$ is a set of restrictions,
denotes a class of all languages recognized by certain
$(a,b)$-QIP systems with the restrictions specified by $\pair{\RR}$.
In addition, we define $\qip(\pair{\RR})$ as the union  $\bigcup_{\epsilon>0}\qip_{1/2+\epsilon,1/2+\epsilon}(\pair{\RR})$.
In this paper, we shall focus our attention on the following three basic
restrictions $\pair{\RR}$:
$\pair{1qfa}$ (\ie 1qfa verifiers), $\pair{2qfa}$ (\ie 2qfa verifiers), and
$\pair{poly\mbox{-}time}$ (\ie expected polynomial running time). As an example,  $\qip(2qfa,poly\mbox{-}time)$ denotes the language class
defined by QIP systems with expected polynomial-time 2qfa verifiers.
Other types of restrictions will be discussed in later sections.

\section{What if Provers Behave Classically?}
\label{sec:classical-prover}

To promote a better understanding of the roles of provers in our QIP systems described  in Section \ref{sec:basic-def}, we shall examine a variant of those  systems. Recall that, in Dwork-Stockmeyer IP systems \cite{DS92}, mighty provers are in essence probabilistic machines that probabilistically select messages to send to verifiers.
As noted in \cite{DS92}, it is possible to reduce those provers to deterministic machines without compromising the language recognition power of the corresponding IP systems.

Naturally, we can ask whether or not standard {\lq\lq}quantum{\rq\rq} provers in our QIP systems can be replaced by significantly weaker machines. Among many candidates for weak machines, we consider machines that operate only unitary operators whose entries are all limited to $0$ and $1$. Significance of such operators is that, using a semi-infinite private tape, those restricted operators essentially make their corresponding provers reduce
to merely deterministic machines. Certainly, a real-life  implementation of such restricted operators could be much simpler and easier than implementing arbitrary unitary operators.
For those reasons, we call
a prover {\em classical}\footnote{In a strict sense, a more exact analogy to deterministic prover may demand that even a communication cell behaves classically.}
if the prover's move is dictated by unitary operators whose entries are $0$s and $1$s.

In comparison, we refer to the original provers (described in Section \ref{sec:basic-def}) as {\em quantum provers}. Remember that classical provers are still quantum provers.
Hereafter, the restriction $\pair{c\mbox{-}prover}$ indicates that all provers behave
classically as defined above. In our QIP systems, classical provers may play an essentially different role from quantum provers.

Let us examine the power of classical-prover QIP systems when verifiers are limited to  1qfa's. It is not difficult to prove that  $\mathrm{1QFA}\subseteq\qip(1qfa,c\mbox{-}prover)$ by forcing provers to unalter the communication cell at any step. However, it is not clear whether $\qip(1qfa,c\mbox{-}prover)$ coincides with $\qip(1qfa)$.  In what follows, we shall demonstrate that $\qip(1qfa,c\mbox{-}prover)$ actually contains $\qip(1qfa)$.

\begin{proposition}\label{public-cprover}
$\qip(1qfa) \subseteq \qip(1qfa,c\mbox{-}prover)$.
\end{proposition}

\begin{proof}
It was shown in \cite[Proposition 4.2]{NY09} that $\mathrm{REG} \subseteq \qip_{1,1}(1qfa)$.
In a similar way, we can prove that $\mathrm{REG} \subseteq \qip(1qfa,c\mbox{-}prover)$.
Since $\qip(1qfa)=\mathrm{REG}$ \cite[Theorem 4.1]{NY09}, we obtain the desired containment  $\qip(1qfa)\subseteq \qip(1qfa,c\mbox{-}prover)$.
\end{proof}

\sloppy
Next, we shall examine the case of 2qfa verifiers.
Unlike the 1qfa-verifier case, any containment between $\qip(2qfa)$ and $\qip(2qfa,c\mbox{-}prover)$ is currently unknown;
nonetheless, we can verify that
$\qip(2qfa,poly\mbox{-}time,c\mbox{-}prover)$ contains $\mathrm{2QFA}(poly\mbox{-}time)$. Therefore, the proper inclusion $\mathrm{REG} \subsetneqq \qip(2qfa,poly\mbox{-}time,c\mbox{-}prover)$
follows as a direct consequence of the result $\mathrm{REG}\subsetneqq\mathrm{2QFA}(poly\mbox{-}time)$ \cite{KW97}.
The following theorem further strengthens this separation to Arthur-Merlin proof systems.

\begin{theorem}\label{am-cprover}
\begin{enumerate}
  \setlength{\topsep}{0mm}%
  \setlength{\itemsep}{1mm}
  \setlength{\parskip}{0cm}%
  
\item  $\am(2pfa)\subsetneqq \qip(2qfa,c\mbox{-}prover)$.

\item $\am(2pfa,poly\mbox{-}time)\subsetneqq \qip(2qfa,poly\mbox{-}time,c\mbox{-}prover)\nsubseteq \am(2pfa)$.
\end{enumerate}
\end{theorem}

\begin{proof}
In a quantum-prover model, it was shown in \cite[Lemma 5.2]{NY09} that the language  $Pal_{\#}=\{x\# x^R\mid x\in\{0,1\}^*\}$ (where $x^R$ is $x$ in the reverse order) of {\em marked even-length palindromes} belongs to $\qip(2qfa,poly\mbox{-}time)$.
By a careful examination of the proof, we find that
the same proof works for classical provers.
This fact immediately places $Pal_{\#}$ into
$\qip(2qfa,poly\mbox{-}time,c\mbox{-}prover)$.
Hence, the non-containment of $\qip(2qfa,poly\mbox{-}time,c\mbox{-}prover)$ inside $\am(2pfa)$
naturally follows because $Pal_{\#}$ is located
outside of $\am(2pfa)$ \cite{DS92}.
This separation further leads to the difference between $\am(2pfa)$ and $\qip(2qfa,c\mbox{-}prover)$. Similarly, we can show the difference between $\am(2pfa,poly\mbox{-}time)$ and
$\qip(2qfa,poly\mbox{-}time,c\mbox{-}prover)$ since, if they are equal, we obtain $\qip(2qfa,poly\mbox{-}time,c\mbox{-}prover) \subseteq \am(2pfa,poly\mbox{-}time) \subseteq
\am(2pfa)$, contradicting the earlier non-containment.

To complete the proof, we shall prove that $\am(2pfa)\subseteq \qip(2qfa,c\mbox{-}prover)$. Since the proof that begins below
works for any time-bounded model, we also obtain the remaining claim that $\am(2pfa,poly\mbox{-}time) \subseteq \qip(2qfa,poly\mbox{-}time,c\mbox{-}prover)$.

Let $L$ be any language in $\am(2pfa)$ over alphabet $\Sigma$. We want to show that $L$ is also in $\qip(2qfa,c\mbox{-}prover)$.
The important starting point is a fact that $L$ can be recognized as follows
by special finite automata $M$ called  {\em 2npfa's} \cite{CHPW98} that make probabilistic moves and nondeterministic moves in turn. If $x\in L$, then there exists a series of nondeterministic choices by which $M$ halts in accepting states with probability at least $1-\epsilon$; otherwise, for every series of nondeterministic choices, $M$ halts in rejecting states with probability at least $1-\epsilon$, where $\epsilon$ is an appropriate constant in $[0,1/2)$.

Now, we take a 2npfa $M=(Q,\check{\Sigma},\delta_M,q_0,Q_{acc},Q_{rej})$
with nondeterministic states and probabilistic states that recognizes $L$ with error probability at most $\epsilon$,
where $0\leq \epsilon < 1/2$ and $Q$ is made up of four disjoint sets $Q_R$ (a  set of probabilistic states), $Q_N$ (a set of nondeterministic states),
$Q_{acc}$ (a set of accepting states) and $Q_{rej}$ (a set of rejecting states).
To simplify our proof, we force $M$ to satisfy the following two extra conditions:
(i) $M$'s tape head does not stay still at any step and (ii) whenever $M$ tosses a fair coin, the tape head moves only to the right.
It is not difficult to modify any 2npfa to meet those two conditions.
Let us recall the transitions of 2npfa's from \cite{CHPW98}. Given any $(p,\sigma,q,d)\in Q\times\check{\Sigma}\times Q\times \{\pm 1\}$,
a transition $\delta_M(p,\sigma,q,d)$ takes a value in $\{0,1/2,1\}$ in the following manner.
(1) When $p\in Q_R$, there exist two distinct states $p_0,p_1\in Q$ for which  $\delta_M(p,\sigma,q,d)=1/2$ holds if  $(q,d)\in\{(p_0,1),(p_1,1)\}$,
and $\delta_M(p,\sigma,q,d)=0$ otherwise. This indicates that $M$ changes its  inner state from $p$ to each of $p_0$ and $p_1$ with equal probability $1/2$
and moves its tape head rightward (by Condition (ii)).
(2) In the case where $p\in Q_N$, there are $m$ ($\in\nat^{+}$) pairs
$(p_1,d_1),(p_2,d_2),\ldots,(p_m,d_m)\in Q\times \{\pm1\}$ for which
$\delta_M(p,\sigma,q,d)=1$ holds if $(q,d)\in \{(p_1,d_1),\ldots,(p_m,d_m)\}$,
and $\delta_M(p,\sigma,q,d)=0$ otherwise. This means that $M$ chooses an index $i  \in[1,m]_{\integer}$ nondeterministically, changes the inner state from $p$ to $p_i$,
and moves the tape head to the left ($d_i=-1$) or to the right ($d_i=1$).
Notice that deterministic moves are always treated as a special case of nondeterministic moves.
Based on the above machine $M$, we shall construct the desired QIP system $(P,V)$
with classical prover $P$ for $L$.

Let $x$ be any input string of length $n$. Let $Q'= Q\cup \{\hat{p}\mid p\in Q\}$ be a set of inner states and let
$\Gamma=(Q'\times\{\pm1\})\cup\{\#,\kappa\}$ be a communication alphabet,
where $\hat{p}$ is a new inner state associated with $p$ and $\kappa$ is a fresh non-blank symbol.
The verifier $V$ carries out the procedure that follows $\delta_M$, by which $V$ simulates $M$ step by step.

Let us consider any step at which $M$ tosses a fair coin in probabilistic state $p$
by applying a transition $\delta_M(p,\sigma,p_0,1)=\delta_M(p,\sigma,p_1,1)=1/2$ for two distinct states $p_0,p_1\in Q$.
The verifier $V$ checks whether $\#$ is in the communication cell.
Unless this is the case, $V$ rejects $x$ immediately; otherwise, $V$ makes the corresponding $(Q\times\Gamma)$-transition
$V_\sigma\qubit{p}\qubit{\#} = \frac{1}{\sqrt{2}}(\qubit{p_0}\qubit{(p,1)}+ \qubit{p_1}\qubit{(p,1)})$ with $D(p_0,(p,1))= D(p_1,(p,1)) = +1$.
The verifier expects a prover to erase the symbol $(p,1)$
in the communication cell
by overwriting it with $\#$. This erasure of symbols
guarantees $V$'s move to be unitary.

Next, let us consider any step at which $M$ makes a nondeterministic choice $(p,d)$ in state $p\in Q_{N}$, namely, $\delta_M(p,\sigma,q,d)=1$, where $m\in\nat^+$ and $(q,d)\in\{(p_1,d_1),\ldots,(p_m,d_m)\}\subseteq Q\times\{\pm1\}$.
For this $M$'s step, $V$ needs two steps to simulate it.
The verifier $V$ enters a rejecting inner state immediately
unless the communication cell contains $\#$. Now, assume that $\#$ is in the communication cell.
Without moving its tape head, $V$ first sends the designated symbol $\kappa$ to a prover,
requesting a pair $(p',d')$ in $Q\times\{\pm 1\}$ to return.
This is done by the special $(Q\times\Gamma)$-transition $V_\sigma\qubit{p}\qubit{\#}=\qubit{\hat{p}}\qubit{\kappa}$ with $D(\hat{p},\kappa)=0$.
The verifier forces a prover to return a valid form of
nondeterministic choice (namely, $\delta_M(p,\sigma,p',d')=1$)
by entering a rejecting inner state whenever the prover writes any other symbol.
Once $V$ receives a valid pair $(p_i,d_i)$,
he makes the $(Q\times\Gamma)$-transition $V_\sigma\qubit{\hat{p}}\qubit{(p_i,d_i)} = \qubit{p_i}\qubit{(\hat{p},d_i)}$ with $D(p_i,(\hat{p},d_i))=d_i$
and expects a prover to erase the communication symbol $(\hat{p},d_i)$.

The honest prover $P$ must blank out the communication cell at the end of
every simulation step of $V$ and, in request of $V$ with the symbol $\kappa$,  $P$ returns a {\lq\lq}correct{\rq\rq} nondeterministic
choice to $V$ (if any).
Assuming $x\in L$, there are a series of nondeterministic choices along which $M$ accepts $x$ with probability at least $1-\epsilon$. Since
the honest prover $P$ sends such a series step by step, $P$ can guide $V$ to make correct nondeterministic choices. Moreover, $P$ allows $V$ to simulate correctly $M$'s probabilistic moves by erasing $V$'s communication symbols. Hence, $V$ successfully reaches $M$'s outcomes
with the same error probability, and thus the protocol $(P,V)$ accepts $x$ with probability at least $1-\epsilon$.

Next, consider the case where $x\not\in L$. Notice that no matter how nondeterministic choices are made, $M$ rejects $x$ with probability at least $1-\epsilon$. Take a dishonest classical prover $P^*$ that maximizes the acceptance probability of $V$ on $x$.
This particular prover $P^*$ must clear out the communication cell whenever $V$ asks him to do so since,
otherwise, $V$ immediately rejects $x$ and thus lowers the acceptance probability, a contradiction against the choice of $P^*$.
Since $P^*$ is classical, all the computation paths of $V$
have nonnegative amplitudes, which cause only non-destructive interference.
This indicates that $P^*$ cannot annihilate any existing computation path of $V$.
On request for a nondeterministic choice, $P^*$ must return any one of valid nondeterministic choices since, otherwise, $V$ rejects immediately.
With a series of nondeterministic choices of $P^*$, if $V$ rejects $x$ with probability less than $1-\epsilon$,
then our simulation implies that $M$ also rejects $x$ with probability less than $1-\epsilon$.
This is a contradiction against our assumption. Hence, $V$ must reject $x$ with probability at least $1-\epsilon$.
Therefore, $(P,V)$ is a classical-prover $(1-\epsilon,1-\epsilon)$-QIP system for $L$.
\end{proof}

In the above proof, we cannot replace classical provers
by quantum provers, mainly because a certain quantum prover may fool the constructed verifier by (i) returning a superposition of nondeterministic choices instead of choosing one of the two nondeterministic  choices and
(ii) using negative amplitudes to make the verifier's quantum simulation destructive.

In the end of this section, we shall present a QIP protocol using classical provers for the non-regular language $Center =\{x1y\mid x,y\in\{0,1\}^*, |x|=|y|\}$,
which is known to be in $\am(2pfa)$ but not in $\am(2pfa,poly\mbox{-}time)$ \cite{DS92}.
In the QIP protocol described in the next proof, an honest prover signals the location of the center bit of a given input
and then a verifier tests the correctness of the location by employing the {\em quantum Fourier transform} (or {\em QFT}, in short)
in a fashion similar to \cite{KW97}. An {\em interaction} in a QIP protocol  constitutes a verifier's transition, a projective measurement, and a prover's move.

\begin{lemma}\label{qip-vs-am}
For any constant $\epsilon\in(0,1)$, $Center\in \qip_{1,1-\epsilon}(2qfa,poly\mbox{-}time,c\mbox{-}prover)$.
\end{lemma}

\begin{proof}
Let $\epsilon$ be any error bound in the real interval $(0,1)$ and set $N=\ceilings{1/\epsilon}$.
In what follows, we shall define
the desired QIP protocol that witnesses the membership of $Center$ to
$\qip_{1,1-\epsilon}(2qfa,poly\mbox{-}time,c\mbox{-}prover)$.
Let $\Sigma=\{0,1\}$ be our input alphabet and let $\Gamma=\{\#,1\}$ be our communication alphabet. Our QIP protocol $(P,V)$
comprises four phases. The formal description of the behavior of $V$ is given in Table \ref{table:center} using $(Q\times\Gamma)$-transitions $\{V_{\sigma}\}_{\sigma\in\check{\Sigma}}$. Let $x$ be an arbitrary input.

\begin{table}[ht]
\bs\begin{center}
{\small
\begin{tabular}{|ll|}\hline
$V_{\cent}\qubit{q_0}\qubit{\#} = \qubit{q_0}\qubit{\#}$
& $V_{\$}\qubit{q_0}\qubit{\#} =\qubit{q_{rej,0}}\qubit{\#}$  \\

$V_{\cent}\qubit{q_2}\qubit{1}=\qubit{q_{rej,0}}\qubit{\#}$
& $V_{\$}\qubit{q_0}\qubit{1} =\qubit{q_{rej,1}}\qubit{\#}$ \\

$V_{\cent}\qubit{q_2}\qubit{\#}=\qubit{q_3}\qubit{\#}$
& $V_{\$}\qubit{q_1}\qubit{\#} =\qubit{q_2}\qubit{\#}$ \\

$V_{\cent}\qubit{s_{j,0}}\qubit{1} = \frac{1}{\sqrt{N}}\sum_{l=1}^N \mathrm{exp}(2\pi \imath jl/{N})
\qubit{t_{l}}\qubit{\#}$ ($1\leq j\leq N$)
& $V_{\$}\qubit{q_1}\qubit{1} =\qubit{q_{rej,1}}\qubit{1}$ \\

&  $V_{\$}\qubit{r_{j,0}}\qubit{1} = \qubit{s'_{j,0}}\qubit{1}$  ($1\leq j\leq N$) \\

& $V_{\$}\qubit{s'_{j,0}}\qubit{1} = \qubit{s_{j,0}}\qubit{1}$  ($1\leq j\leq N$) \\

$V_b\qubit{q_0}\qubit{\#}=\qubit{q_1}\qubit{\#}$
& $V_b\qubit{q_1}\qubit{\#}=\qubit{q_0}\qubit{\#}$ \\

$V_b\qubit{q_0}\qubit{1} = \qubit{q_{rej,0}}\qubit{\#}$
& $V_b\qubit{q_1}\qubit{1} = \qubit{q_{rej,0}}\qubit{1}$ \\

$V_b\qubit{q_2}\qubit{\#} = \qubit{q_2}\qubit{\#}$ & $V_b\qubit{q_2}\qubit{1}=\qubit{q_{rej,1}}\qubit{1}$
 \\

$V_b\qubit{q_3}\qubit{\#}=\qubit{q_3}\qubit{\#}$ & \\

$V_1\qubit{q_3}\qubit{1} = \frac{1}{\sqrt{N}}\sum_{j=1}^N \qubit{r_{j,0}}\qubit{\#}$
& $V_0\qubit{q_3}\qubit{1} = \qubit{q_{rej,-1}}\qubit{\#}$ \\

$V_b\qubit{r_{j,0}}\qubit{1} = \qubit{r'_{j,N-j}}\qubit{1}$ ($1\leq j\leq N-1$)
& $V_b\qubit{r_{j,0}}\qubit{\#}=\qubit{q_{rej,j}}\qubit{1}$ ($1\leq j\leq N-1$) \\

$V_b\qubit{r_{j,k}}\qubit{1} = \qubit{r'_{j,k}}\qubit{1}$ ($1\leq k\leq N-j$, $1\leq j\leq N-1$)
& $V_b\qubit{r_{N,0}}\qubit{1} = \qubit{r_{N,0}}\qubit{1}$ \\

$V_b\qubit{r'_{j,k}}\qubit{1} = \qubit{r_{j,k-1}}\qubit{1}$ ($2\leq k\leq N-j$, $1\leq j\leq N-1$)  & \\

$V_b\qubit{r'_{j,1}}\qubit{1} = \qubit{r_{j,0}}\qubit{1}$, ($1\leq j\leq N-1$) &  \\

$V_b\qubit{s_{j,k}}\qubit{1}=\qubit{s_{j,k-1}}\qubit{1}$ ($2\leq k\leq j$, $1\leq j\leq N$)
& $V_b\qubit{s_{j,0}}\qubit{1} = \qubit{s_{j,j}}\qubit{1}$ ($1\leq j\leq N$) \\

$V_b\qubit{s_{j,1}}\qubit{1}=\qubit{s_{j,0}}\qubit{1}$ ($1\leq j\leq N$)
& $V_b\qubit{s_{j,0}}\qubit{\#}=\qubit{q_{rej,N+j}}\qubit{\#}$ ($1\leq j\leq N$)\\

& \\

$D(q_0)=D(q_1)=D(q_3)=1$, $D(q_2)=-1$ & $D(r_{j,0})=1$ ($1\leq j\leq N$) \\
$D(r_{j,k}) =D(r'_{j,k})= 0$ ($1\leq j\leq N-1$, $k\neq 0$)
& $D(s_{j,0})= -1$ ($1\leq j\leq N$) \\
 $D(s_{j,k})= D(s'_{j,0}) = 0$ ($1\leq j\leq N$, $k\neq 0$)  & $D(t_j)=0$ ($1\leq j\leq N$) \\
\hline
\end{tabular}
}
\caption{{\small $(Q\times\Gamma)$-transitions $\{V_{\sigma}\}_{\sigma\in\check{\Sigma}}$ of $V$ for $Center$ with $b\in\{0,1\}$.
Here, $t_N$ is a unique accepting inner state,
while $q_{rej,j}$ ($-1\leq j\leq 2N-1$) and $t_l$ ($1\leq l<N$) are all rejecting inner states.
The table excludes obvious transitions to rejecting inner states
when a prover changes
the communication symbol $1$ to $\#$ during the third and fourth phases.}}\label{table:center}
\end{center}
\end{table}

1) In the first phase, the verifier $V$ checks whether $|x|$ is odd by moving the tape head toward $\$$
together with switching two inner states $q_0$ and $q_1$.
To make deterministic moves during this phase,
$V$ forces a prover to return only the blank symbol $\#$ at any step by entering a rejecting state whenever the prover deceptively
sends back non-blank symbols.
When $|x|$ is odd, $V$ enters the inner
state $q_3$ after moving its tape head
back to $\cent$.
Hereafter, we consider only the case where the input $x$ has an
odd length.

2) In the second phase, $V$ moves its tape head rightward by sending
$\#$ to a prover
until $V$ receives $1$ from the prover.
Receiving $1$ from the prover, $V$ willingly rejects $x$ unless its tape
head is currently scanning $1$ on the input tape.
Otherwise, the third phase starts. During the third and fourth phases,
whenever the prover changes the communication symbol $1$ to $\#$, $V$ immediately rejects the input.

3) Assume that the tape head is now scanning $1$ on the input tape. In the third phase, the computation splits into $N$ parallel branches by applying $V_1\qubit{q_3}\qubit{1}$.
This step is called the {\em first split} and  it
generates the $N$ distinct inner states $r_{1,0},r_{2,0},\ldots,r_{N,0}$ with equal amplitudes $1/\sqrt{N}$.
The tape head then moves deterministically toward $\$$ in the following manner:
along the $j$th computation path ($1\leq j\leq N$) associated with the inner state $r_{j,0}$,
the tape head idles for $2(N-j)$ steps in each tape cell before moving to the next one by changing inner states as
\[
r_{j,0} \xrightarrow{1} r'_{j,N-j} \xrightarrow{0} r_{j,N-j-1} \xrightarrow{0} r'_{j,N-j-1} \xrightarrow{0} r_{j,N-j-2} \xrightarrow{0} \cdots \xrightarrow{0} r'_{j,1} \xrightarrow{0} r_{j,0},
\]
where each number over arrows indicates the direction of the tape head.
When the tape head reaches $\$$, it steps back one cell by applying $V_{\dollar}\qubit{r_{j,0}}\qubit{1}$ and $V_{\dollar}\qubit{s'_{j,0}}\qubit{1}$,  and then starts the fourth phase.

4) During the fourth phase, the tape head along the $j$th computation path keeps moving leftward by idling in each cell for $j$ steps, changing inner states as
\[
s_{j,0} \xrightarrow{0} s_{j,j} \xrightarrow{0} s_{j,j-1} \xrightarrow{0} s_{j,j-2} \xrightarrow{0}  \cdots \xrightarrow{0} s_{j,1} \xrightarrow{-1} s_{j,0}
\]
until the tape head reaches $\cent$.
At $\cent$, the computation splits again into $N$ parallel branches (called the {\em second split}) by applying the QFT $V_{\cent}\ket{s_{j,0}}\qubit{1}$,  yielding either the accepting inner state $t_N$ or one of the rejecting inner states in
$\{t_j\mid 1\leq j< N\}$.

{}From Table \ref{table:center}, it is not difficult to check that $V$ is indeed
well-formed (namely, $U^x_{\delta}$ is unitary for every $x\in\Sigma^*$).
The honest prover $P$ should return $1$ exactly at the time when $V$ scans the center bit of an input string
and at the time when $V$ sends $\#$ to $P$
during the third and fourth phases.
At any other step, $P$ should apply the identity operator.

Now, we shall check the completeness and soundness of the obtained
QIP system $(P,V)$ for $Center$.
First, consider a positive instance $x$, which is of the form $y1z$ for certain strings $y$ and $z$ of the same length,
say, $n$. Since the honest prover $P$ signals just before
$V$ reads the center bit $1$ of $x$,
the first split given by $V_1\qubit{q_3}\qubit{1} = \frac{1}{\sqrt{N}}\sum_{j=1}^{N}\qubit{r_{j,0}}\qubit{\#}$ occurs at the middle of $x$ during the third phase (more precisely, exactly after $n$ steps of $V$ from the start of the second phase) after reading $\cent y1$.
Along the $j$th computation path ($1\leq j\leq N$) associated with
the inner state $r_{j,0}$ chosen at the first split,
$V$ idles for $2n(N-j)$ steps while reading $z$
and also idles
for $2nj$ steps while reading the whole input.
Overall, the idling time elapses for the duration of $2n(N-j)+2nj=2nN$, which is independent of $j$.
Hence, all the $N$ computation paths created at the aforementioned first split must have the same length, and thus the superposition of global configurations prior to the second split becomes $\frac{1}{\sqrt{N}}\sum_{j=1}^{N}\qubit{s_{j,0},0}\qubit{\#}\qubit{\Psi}$ for an appropriate quantum state $\qubit{\Psi}$ in the Hilbert space $\PP$ associated with the prover's private tape.
The QFT given by the transition $V_{\cent}\qubit{s_{j,0}}\qubit{1} = \frac{1}{\sqrt{N}}\sum_{l=1}^{N}\exp(2\pi \imath jl/N)\qubit{t_{l}}\qubit{\#}$ makes all the global configurations converge to the verifier's visible accepting configuration $\qubit{t_N}\qubit{\#}$; that is, $\frac{1}{\sqrt{N}}\sum_{l=1}^{N} \left( \frac{1}{\sqrt{N}}\sum_{j=1}^{N} \exp(2\pi \imath  jl/N)\right) \qubit{t_l,0}\qubit{\#}\qubit{\Psi}$, which equals $\qubit{t_N,0}\qubit{\#}\qubit{\Psi}$.
Therefore, $V$ accepts $x$ with probability $1$.

On the contrary, suppose that $x$ is a negative instance of the form
$x=y0z$ with $|y|=|z|=n$. Consider the second, third, and fourth phases.
To minimize the rejection probability, a dishonest prover $P^*$ must
send the symbol $1$ just before $V$ scans $1$ on the input tape during
the second phase and then $P^*$ must maintain $1$ because, otherwise,
$V$ immediately rejects $x$. Note that
there is no way for classical provers to pass both $1$ and $\#$ in a form of superposition to deceive the verifier.
Let us assume that the $e$th symbol of $x$ is 1 and $P^*$ sends
$1$ during the $e$th interaction, where $1\leq e\leq 2n+1$.
Obviously, $e\neq n+1$ follows because the center bit of $x$ is $0$.
Consider the first split caused by applying $V_{1}\qubit{q_3}\qubit{1} = \frac{1}{\sqrt{N}}\sum_{j=1}^{N}\qubit{r_{j,0}}\qubit{\#}$.
For each index $j\in[1,N]_{\integer}$, let $p_{j}$ be the computation path following the $j$th branch that starts with the inner state $r_{j,0}$
generated at the first split.
Along this computation path $p_{j}$, the idling time totals $2(|x|-e)(N-j)+2nj=2(n+1-e)(N-j)+2nN$.
Since $1\leq j\leq N$, two computation paths $p_{j}$ and $p_{j'}$ for any distinct values $j$ and $j'$ must have different lengths.
Just before the second split, along the $j$th computation path, we obtain a quantum state $\frac{1}{\sqrt{N}}\qubit{s_{j,0},0}\qubit{\#}\qubit{\Psi_j}+ \ket{\Delta_j}$, where $\ket{\Delta_j}$ does not contain $\qubit{s_{j,0}}$. At the second split,
the QFT further generates $N$ parallel branches  $\frac{1}{N}\sum_{j=1}^{N}\exp(2\pi\imath jl/N)\qubit{t_l,0}\qubit{\#}\qubit{\Psi_j}+ \ket{\Delta'_j}$, which equals
$\frac{1}{N}\qubit{t_N,0}\qubit{\#}\qubit{\Psi_j} + \ket{\Delta''_j}$, where $\ket{\Delta'_j}$ is obtained from $\ket{\Delta_j}$ by the QFT and $\ket{\Delta''}$ is  an appropriate quantum state not containing $\qubit{t_N}$. Thus, at most
one of the computation paths can reach $\qubit{t_N,0}\qubit{\#}$. Hence, the probability of $V$ reaching such an acceptance configuration is no more than $1/N^2$.
Since there are $N$ computation paths $\{p_j\}_{1\leq j\leq N}$ generated at the first split, the overall acceptance probability is at most $N\times(1/N^2)=1/N$. Since $V$'s computation paths always end with certain halting states, it follows that $V$ rejects $x$ with probability $\geq 1- 1/N \geq 1-\epsilon$.
\end{proof}

\section{What If a Verifier Reveals Private Information?}\label{sec:public-coin}

In Dwork-Stockmeyer IP systems \cite{DS92}, the prover's view of the verifier's computation is limited to a small window (\ie a communication cell) and the strength of a prover's strategy hinges on the amount of the information
that a verifier is willing to reveal to the prover through this window.
Let us consider a situation, in their IP system, that a verifier always unalters the communication cell.
Since the behavior of a 2pfa verifier depends on not only messages from a prover but also its internal random choices (or its coin flips), no prover can gain more than the information on the number of the verifier's moves, and therefore any prover knows little of the
verifier's actual configurations.

In Babai's {\em Arthur-Merlin proof systems} \cite{Bab85} (also known as {\lq\lq}public-coin{\rq\rq} IP systems \cite{GS86}), on the contrary,
the verifier must always pass the information on his next move resulting from   his internal random choices, and such information suffices for
the mighty prover to keep track of the verifier's configurations.
Dwork and Stockmeyer \cite{DS92} defined $\am(\pair{restriction})$ as a variant of their original IP systems by requiring their verifiers to publicly reveal  next inner states and tape head directions determined by internal coin flips.

Here, we shall consider a straightforward quantum analogy of the above public-coin IP systems and investigate their language recognition power. In our QIP system, we demand the verifier to reveal through   the communication cell his choice of non-halting inner state as well as his tape head direction {\em at every step}.
Formally, we define a {\em public QIP system} as follows, whereas we sometimes call the original QIP systems defined in Section \ref{sec:basic-def} {\em private QIP systems}  in comparison.

\begin{definition}\label{ourpublic}
A 2qfa-verifier QIP system $(P,V)$ is called {\em public}\footnote{As another variant of public QIP system, we may require  $U_{\delta}^{x}\qubit{q,k,\gamma}$ to satisfy the same equality only for non-halting states $q'$, instead of $q$. See \cite{NY04} for more details.}
if the verifier's linear operators $\{U_{\delta}^{x}\}_{x\in\Sigma^*}$
induced by $\delta$ satisfy the following {\em publicness condition}:
for any tuple $(x,q,k,\gamma)$, if $q$ is a non-halting state, then  $U_\delta^x\qubit{q,k,\gamma}$ must be of the form
$\sum_{q',\xi,d} \delta(q,x_{(k)},\gamma,q',\xi,d)
\qubit{q',k+d\ (\mbox{mod }|x|+2),\xi}$, where $\xi=(q',d)$, $x_{(0)}=\cent$, and $x_{(|x|+1)}=\dollar$.
\end{definition}

In particular, when the verifier $V$ is a 1qfa, we omit the
information on tape head direction $d$ from the communication symbol
$\xi=(q',d)$ in Definition \ref{ourpublic} since $V$ always moves its tape head to the right (\ie $d=+1$) and the information on $d$ is obviously redundant.
To emphasize the {\lq\lq}publicness{\rq\rq} of this new system,
we use the specific notation $\pair{public}$.
For instance, $\qip(2qfa,public)$ indicates a collection of all languages recognized by public QIP systems with 2qfa verifiers.
By direct analogy with $\am(2pfa)$, however, we might possibly write $\mathrm{QAM}(2qfa)$ for $\qip(2qfa,public)$.  In Definition \ref{ourpublic}, since there is no restriction on provers, all public QIP systems with 1qfa verifiers are naturally private QIP systems with the same verifiers. It therefore holds that, for example,
$\qip(1qfa,public) \subseteq \qip(1qfa)$ and $\qip(2qfa,public) \subseteq \qip(2qfa)$.

We shall further demonstrate the power of public QIP systems. Now, we shall concentrate on the language class $\qip(1qfa,public)$.
Unlike $\mathrm{1QFA}\subseteq\qip(1qfa)$, the containment $\mathrm{1QFA}\subseteq \qip(1qfa,public)$,
which seems to hold naturally at a quick glance, is still unknown.
The difficulty of proving this containment is caused by the publicness condition of the public QIP systems.
Because the verifier must announce its next move to a prover, he unintentionally helps
the prover make his local system entangled with the prover's local system; however, we do not know how to get rid of this type of entanglement.

Despite the publicness condition, we can still show that the power of $\qip(1qfa,public)$ is well beyond $\mathrm{1QFA}$.
Let us consider the language
$Zero = \{w0\mid w\in\{0,1\}^*\}$ in $\qip(1qfa,public)$, which
is known to reside outside of $\mathrm{1QFA}$ \cite{KW97}. In the next lemma, we shall prove that $Zero$ has a public QIP system with 1qfa verifiers.

\begin{lemma}\label{zero-public-1qfa}
$Zero\in  \qip_{1,1}(1qfa,public)$.
\end{lemma}

The following proof exploits the prover's ability to inform the location of the rightmost bit $0$ of an instance in $Zero$. To simplify the description of 1qfa verifiers $V=(Q,\check{\Sigma},\delta,q_0,Q_{acc},Q_{rej})$, for each $q\in Q$, we intend to abbreviate the communication symbol $(q,1)$ as $q$, because  $V$'s tape head direction is always $+1$.

\begin{proofof}{Lemma \ref{zero-public-1qfa}}
Let us show that $Zero$ has a public $(1,1)$-QIP system $(P,V)$ with
1qfa verifier $V$.
To describe the desired protocol $(P,V)$, let $\Sigma=\{0,1\}$ be its input alphabet and let
$Q_{non}=\{q_0,q_1\}$, $Q_{acc}=\{q_{acc,0},q_{acc,1},q_{acc,-1}\}$ and $Q_{rej}=\{q_{rej,i},q'_{rej,i}\mid i\in\{0,\pm1\}\}$
be respectively the sets of  non-halting inner states, of accepting inner states, and of rejecting inner states of $V$.
Using the above-mentioned abbreviation, our communication alphabet $\Gamma$ can be defined as $\{\#,q_0,q_1\}$.

The protocol of $V$ is described in the following. See Table \ref{table:public-Lzero} for the formal description of $V$'s $(Q\times\Gamma)$-transitions.
Let $x=yb$ be any input string, where $b\in\{0,1\}$.
The verifier $V$ stays in the initial state $q_0$ by publicly announcing $q_0$  (\ie sending the communication symbol $q_0$ to a prover)
until the prover returns $\#$. Whenever $V$ receives $\#$, he immediately rejects $x$ by entering $q_{rej,-1}$ (after applying either $V_{\dollar}\qubit{q_0}\qubit{\#}$ or $V_1\qubit{q_0}\qubit{\#}$)
if its current scanning symbol is different from $0$.
On the contrary, if $V$ is scanning $0$, then he waits for the next tape symbol by entering $q_1$. If the next symbol is $\dollar$, then he accepts $x$ after applying $V_{\dollar}\qubit{q_1}\qubit{q_1}$; otherwise, he rejects $x$ by entering $q'_{rej,i}$ (after applying either $V_0\qubit{q_1}\qubit{q_i}$ or $V_1\qubit{q_1}\qubit{q_i}$).
Our honest prover $P$ does not alter the communication cell until $V$ reaches the right end of $\cent y$ and $P$ must return $\#$ just before $V$ reads the  symbol $b$ so that $V$ can apply $V_b\qubit{q_1}\qubit{\#}$.

\begin{table}[ht]
\bs\begin{center}
\begin{tabular}{|lll|}\hline
$V_{\cent}\qubit{q_0}\qubit{\#}=\qubit{q_0}\qubit{q_0}$
&  $V_{0}\qubit{q_0}\qubit{\#}=\qubit{q_{1}}\qubit{q_{1}}$  &
$V_{1}\qubit{q_{0}}\qubit{\#} =\qubit{q_{rej,-1}}\qubit{q_{rej,-1}}$ \\

$V_{\cent}\qubit{q_0}\qubit{q_j} = \qubit{q_{rej,j}}\qubit{q_{rej,j}}$
& $V_{0}\qubit{q_0}\qubit{q_0}=\qubit{q_0}\qubit{q_0}$ & $V_{1}\qubit{q_0}\qubit{q_0}=\qubit{q_0}\qubit{q_0}$  \\

$V_{\cent}\qubit{q_1}\qubit{q_i} = \qubit{q'_{rej,i}}\qubit{q'_{rej,i}}$ &
$V_{0}\qubit{q_{0}}\qubit{q_1} =\qubit{q_{rej,1}}\qubit{q_{rej,1}}$ & $V_{1}\qubit{q_{0}}\qubit{q_1} =\qubit{q_{rej,1}}\qubit{q_{rej,1}}$  \\

$V_{\dollar}\qubit{q_0}\qubit{q_i} =\qubit{q_{rej,i}}\qubit{q_{rej,i}}$ &
$V_{0}\qubit{q_1}\qubit{q_i}=\qubit{q'_{rej,i}}\qubit{q'_{rej,i}}$
& $V_{1}\qubit{q_1}\qubit{q_i}=\qubit{q'_{rej,i}}\qubit{q'_{rej,i}}$  \\

$V_{\dollar}\qubit{q_{1}}\qubit{q_i} =\qubit{q_{acc,i}}\qubit{q_{acc,i}}$
& & \\
\hline
\end{tabular}
\caption{{\small $(Q\times\Gamma)$-transitions $\{V_{\sigma}\}_{\sigma\in\check{\Sigma}}$ of $V$ for $Zero$ with $i\in \{0,\pm 1\}$ and $j\in\{0,1\}$. The symbol $q_{-1}$ denotes $\#$.}}\label{table:public-Lzero}
\end{center}
\end{table}

It still remains to prove that $(P,V)$ recognizes $Zero$ with certainty. Consider the case where our input $x$ is of the form $y0$ for a certain string $y$. Since $x$ is in $Zero$, the honest prover $P$ returns $\#$ just after   $V$ reads the rightmost symbol of $\cent y$. This information helps $V$ locate the end of $y$.
Moving its tape head rightward, $V$ confirms that the next scanning symbols are $0\dollar$ and then enters an accepting inner state (either $q_{acc,0}$, $q_{acc,1}$, or $q_{acc,-1}$) with probability $1$. On the contrary, assume that $x=y1$.
Clearly, the best adversary $P^*$ needs to return either $q_0$ or $\#$ (or their superposition).
If $P^*$ keeps returning $q_0$, then $V$ eventually rejects $x$ and increases the rejection probability.
Since $V$'s computation is essentially deterministic, this strategy only decreases the chance of cheating by $P^*$.
To make the best of the adversary's strategy, $P^*$ must return the communication symbol $\#$ just before $V$ scans $0$.
Nonetheless, when $P^*$ returns $\#$, $V$ applies $V_0\qubit{q_0}\qubit{\#}$ and then
applies $V_0\qubit{q_1}\qubit{q_i}$ or $V_1\qubit{q_1}\qubit{q_i}$,
where $i\in\{0,\pm1\}$ and $q_{-1}=\#$. Obviously, this leads to a rejecting inner state of $V$ with certainty.
Therefore, the QIP system $(P,V)$ recognizes $Zero$ with certainty.
\end{proofof}

It follows from Lemma \ref{zero-public-1qfa} that $\qip(1qfa,public)$ is powerful enough to contain certain languages that cannot be recognized by 1qfa's alone.
It is also possible to show that  $\qip(1qfa,public)$ contains all languages recognized by 1qfa's whose transition amplitudes are limited to $\{0,1\}$. Those 1qfa's are known as {\em 1-way (deterministic) reversible finite automaton} ({\em 1rfa}, in short) \cite{AF98}.
For our convenience, let $\mathrm{1RFA}$ denote the collection of all languages recognized by such 1rfa's.
As Ambainis and Freivalds \cite{AF98} showed, $\mathrm{1RFA}$ is characterized exactly as the collection of all languages
that can be recognized by 1qfa's with success probability $\geq 7/9+\epsilon$ for certain constants $\epsilon>0$.

\begin{theorem}\label{QFA-in-public-qfa}
$\mathrm{1RFA}\subsetneqq \qip_{1,1}(1qfa,public)\nsubseteq \mathrm{1QFA}$.
\end{theorem}

\begin{proof}
Firstly, we shall show that $\qip_{1,1}(1qfa,public)$ contains
$\mathrm{1RFA}$.
Take an arbitrary set $L$ recognized by a 1rfa $M=(Q,\check{\Sigma},\delta_{M},q_0,Q_{acc},Q_{rej})$,
where $\delta_{M}$ is a {\em reversible transition function} \cite{AF98} from $Q\times\check{\Sigma}\times Q$ to $\{0,1\}$.
This $\delta_{M}$ satisfies that (i) for any pair $(p,\sigma)\in Q\times\check{\Sigma}$, there exists a unique inner state $q\in Q$  for which $\delta_{M}(p,\sigma,q)=1$ and (ii) for any $(q,\sigma)\in Q\times\check{\Sigma}$, there is at most one $p\in Q$ satisfying $\delta_{M}(p,\sigma,q)=1$.

Henceforth, we shall construct a public $(1,1)$-QIP system $(P,V)$ that {\lq\lq}mimics{\rq\rq} a computation of $M$.
The desired 1qfa verifier $V=(Q',\check{\Sigma},\Gamma,\delta,q_0,Q'_{acc},Q'_{rej})$
behaves as follows. Let $Q'_{acc}= Q_{acc}$, $Q'_{rej} = Q_{rej}\cup \{q_{rej,p,q}\mid p\in Q_{non},q\in Q,p\neq q\}$, and
$Q' = \Gamma = Q\cup Q'_{rej}$, provided that $q_{rej,p,q}$'s are all fresh symbols not in $Q$.
Assume that $V$ is in inner state $p$, scanning symbol $b$ on an input tape.
Whenever $M$ changes its inner state from $p$ to $q$ after scanning $b$,
$V$ does so by revealing its next inner state $q$ to a prover. As soon as $V$ finds that the communication symbol has been altered intentionally
by the prover, $V$ immediately rejects the input. This process forces any prover to unalter the content of the communication cell.
Table \ref{table:public-cprover} gives
a list of $(Q\times\Sigma)$-transitions that induces $V$'s strategy  $\delta$. It is clear from the list that $\delta$ is well-formed  because of the reversibility of $\delta_{M}$ and that
the publicness condition for $V$ is met.
Finally, the honest prover $P$ is a prover who does not alter any communication symbol; that is, $P$ applies only the identity operator at every step.

\begin{table}[ht]
\bs\begin{center}
\begin{tabular}{|l|}\hline
$V_{\cent}\qubit{q_0}\qubit{\#}= \qubit{q}\qubit{q}$ if $\delta_{M}(q_0,\cent,q)=1$ \\

$V_{b}\qubit{p}\qubit{p}= \qubit{q}\qubit{q}$ if $\delta_{M}(p,b,q)=1$ \\

$V_{b}\qubit{p}\qubit{q}=  \qubit{q_{rej,p,q}}\qubit{q_{rej,p,q}}$
if $p\neq q$ and $p\in Q_{non}$ \\

\hline
\end{tabular}
\caption{{\small $(Q\times\Gamma)$-transitions $\{V_{\sigma}\}_{\sigma\in\check{\Sigma}}$ of $V$ for $L$ with $b\in\Sigma\cup\{\dollar\}$ and $p,q\in Q$. All inner states $q_{rej,p,q}$ are rejecting states.}}\label{table:public-cprover}
\end{center}
\end{table}

On input $x\in\Sigma^*$, the QIP system $(P,V)$ accepts $x$ with certainty if $x\in L$, since $V$ exactly simulates $M$ by the help of the honest prover $P$.
Let us consider the opposite case where $x\not\in L$.
It is easy to see that the best strategy for a dishonest
prover $P^*$ is to keep any communication symbol unchanged
because any alteration of the communication symbols causes $V$ to reject $x$ immediately and lowers the acceptance probability of $V$.
Against such a prover $P^*$, $V$ obviously enables to reject $x$ with certainty because, in this case, $V$'s final decision is not influenced
by the communication symbols.
Therefore, $(P,V)$ recognizes $L$ with certainty.
Since $L$ is arbitrary, we obtain the desired containment $\mathrm{1RFA}\subseteq\qip_{1,1}(1qfa,public)$.

Secondly, the separation between $\mathrm{1QFA}$ and $\qip_{1,1}(1qfa,public)$ immediately follows from Lemma \ref{zero-public-1qfa} together with the fact that
$Zero$ is not in $\mathrm{1QFA}$ \cite{AF98}. Moreover, since  $\mathrm{1RFA}\subseteq\mathrm{1QFA}$ and $\qip_{1,1}(1qfa,public) \nsubseteq \mathrm{1QFA}$, we can conclude that $\mathrm{1RFA}\neq \qip(1qfa,public)$.
This completes the proof.
\end{proof}

Next, we shall examine public QIP systems whose verifiers are
2qfa's. Similar to Theorem \ref{am-cprover}(2), we can claim the
following two separations.

\begin{theorem}\label{public-vs-am}
\begin{enumerate}
  \setlength{\topsep}{0mm}%
  \setlength{\itemsep}{1mm}
  \setlength{\parskip}{0cm}%
  
\item $\qip(2qfa,public,poly\mbox{-}time) \nsubseteq\am(2pfa,poly\mbox{-}time)$.

\item $\qip(2qfa,public,poly\mbox{-}time,c\mbox{-}prover) \nsubseteq\am(2pfa,poly\mbox{-}time)$.
\end{enumerate}
\end{theorem}

A language that separates the public QIP systems with 2qfa verifiers
from $\am(2pfa,poly\mbox{-}time)$
is $Upal=\{0^n1^n\mid n\in\nat \}$. Since  $Upal$ resides outside of $\am(2pfa,poly\mbox{-}time)$ \cite{DS92} and $Upal$ belongs to $\mathrm{2QFA}(poly\mbox{-}time)$ \cite{KW97}, the separation $\mathrm{2QFA}(poly\mbox{-}time)\nsubseteq
\am(2qfa,poly\mbox{-}time)$ follows immediately. This separation, however, does not directly imply Theorem \ref{public-vs-am} because, for a technical reason similar to the case of 1qfa verifiers, it is not known whether $\mathrm{2QFA}(poly\mbox{-}time)$
is included in $\qip(2qfa,public,poly\mbox{-}time)$ or even in $\qip(2qfa,public,poly\mbox{-}time,c\mbox{-}prover)$. Therefore, we still need to prove in the next lemma that $Upal$ indeed belongs to both $\qip(2qfa,public,poly\mbox{-}time)$ and $\qip(2qfa,public,poly\mbox{-}time,c\mbox{-}prover)$.

\begin{lemma}\label{upal-public}
\sloppy
For any constant $\epsilon$ in $(0,1]$, $Upal$ belongs to
the intersection $\qip_{1,1-\epsilon}(2qfa,public,poly\mbox{-}time) \cap \qip_{1,1-\epsilon}(2qfa,public,poly\mbox{-}time,c\mbox{-}prover)$.
\end{lemma}

\begin{proof}
In what follows, we shall prove that $Upal$ belongs to $\qip_{1,1-\epsilon}(2qfa,public,poly\mbox{-}time)$. The proof for  $Upal\in \qip_{1,1-\epsilon}(2qfa,public,poly\mbox{-}time,c\mbox{-}prover)$ is similar.
Let $N=\ceilings{1/\epsilon}$. Let us define our public QIP system $(P,V)$.
The honest prover $P$ always applies the identity operation at every step.
The verifier $V$ acts as follows. In the first phase, it determines whether an input $x$ is of the form $0^m1^n$. The rest of the verifier's algorithm is similar in essence to the one given in the proof of Lemma \ref{qip-vs-am}.
In the second phase, $V$ generates $N$ parallel branches with equal amplitude $1/\sqrt{N}$ by entering $N$ different inner states, say,  $r_{1},r_{2},\ldots,r_{N}$.
In the third phase, along the $j$th branch starting with $r_j$ ($j\in[1,N]_\integer$), the tape head idles for $N-j$ steps at each tape cell containing $0$ and idles for $j$ steps at each cell containing $1$ until the tape head finishes reading $1$s.
In the fourth phase, $V$ applies the QFT to collapse all computation paths to a single accepting inner state if $m=n$. Otherwise, all the computation paths do not interfere with each other since the tape head reaches $\dollar$ at different times along different computation paths. During the first and second phases, $V$ publicly reveals the information $(q',d')$ on his next move and then checks whether the prover rewrites it with a different symbol.
To constrain the prover's strategy, $V$ immediately enters a rejecting inner state
if the prover alters the content of the communication cell.

An analysis of the QIP protocol $(P,V)$ for its completeness and soundness conditions is essentially the same as in the proof of Lemma \ref{qip-vs-am}. In conclusion, $Upal$ is in $\qip_{1,1-\epsilon}(2qfa,public,poly\mbox{-}time)$, as requested.
\end{proof}

\section{How Many Interactions are Necessary or Sufficient?}
\label{sec:interaction}

In the previous two sections, despite heavy restrictions on QIP systems,
we have witnessed that quantum interactions between
a prover and a qfa verifier remarkably enhance the qfa's ability to recognize certain types of languages.
Since our basic QIP model forces a verifier to communicate
with a prover {\em at every move}, it is natural to ask whether
such interactions are truly necessary.
To answer this question, we shall remodel QIP systems so that verifiers are allowed to communicate with provers only at the time when the verifiers need any help from the provers. Throughout this section, we shall shed new light on the number of
interactions between a prover and a verifier in those new QIP systems, and we shall carefully examine how many interactions are necessary or sufficient to conduct a given task of language recognition.

\subsection{Interaction-Bounded QIP Systems}\label{sec:new-QIP-model}

To study the number of interactions between a prover and a verifier,
we want to modify our basic QIP systems so that a prover is permitted to
alter a communication symbol in the communication cell exactly when
the verifier asks the prover to do so. To make such a modification, we first
look into the IP systems of Dwork and Stockmeyer \cite{DS92}. In their
system, a verifier is allowed to do computation ``silently'' at
any chosen time with no communication
with a prover; in other words, the verifier interacts with the prover
only when the help of the prover is needed and the prover patiently
awaits for next interactions without conducting any computation.
We interpret the verifier's
silent mode as follows: if the verifier $V$ does not wish to communicate
with the prover, he writes a special communication symbol
in the communication cell to signal the prover that he needs no
help from the prover. Simply, we use the blank symbol $\#$ to condition that
the prover is prohibited to tailor the content of
the communication cell.

We formally introduce a new QIP system, in which no malicious
prover $P$ is permitted to cheat a verifier by willfully tampering with
the symbol $\#$  in the communication cell. Since the verifier is governed by quantum mechanics, if a malicious prover willfully modifies $\#$, the verifier's computation may be significantly hampered and the verifier may have no means to prevent such an action of the prover because of the unitarity requirement of the verifier's strategy $\delta$.
To describe a {\lq\lq}valid and legitimate{\rq\rq} prover $P$
independent of the choice of verifiers, we require the prover's strategy
$P_x=\{U_{P,i}^x\}_{i\in\nat^+}$  acting on the
prover's visible configuration space $\MM\otimes\PP$ on each input $x$ to do nothing (namely, apply the identity operator). To allow a prover $P$
to maintain the unitarity of his strategy $U^{x}_{P,i}$, we also permit the prover to modify his private information $\gamma$ (including a content of the communication cell) when $\gamma$ never appears in an actual computation with non-zero amplitudes. To formulate this condition independent of the verifier,
we need to introduce a series $\{S_i\}_{i\in\nat}$ of elements in $\Delta^{\infty}_{fin}$.
This series $\{S_i\}_{i\in\nat}$ is defined recursively as $S_0=\{\#^{\infty}\}$
and $S_{i}$ ($i\in\nat^{+}$) is the collection of all elements $y\in \Delta^{\infty}_{fin}$
such that, for a certain element $z\in S_{i-1}$ and certain
communication symbols $\sigma,\tau\in\Gamma$,
the superposition $U^{x}_{P,i}\qubit{\sigma}\qubit{z}$ contains
the visible configuration $\qubit{\tau}\qubit{y}$ of {\em non-zero} amplitude, namely, $|\bra{y}\bra{\tau}U_{P,i}^{x}\ket{\sigma}\ket{z}|>0$. Now, our requested
condition is expressed as follows.
\begin{quote}
\begin{itemize}
\item[(*)] For every $i\in\nat^{+}$ and every $y\in S_{i-1}$,
$U_{P,i}^x\qubit{\#}\qubit{y}=\qubit{\#}\qubit{y}$.
\end{itemize}
\end{quote}
Any  prover $P$ who satisfies Condition (*) is succinctly referred to as
{\em committed}.\footnote{There are a number of possible variants, one of which requires that, for every $i\in\nat^{+}$ and for every $y\in\Delta^{\infty}_{fin}$, $U_{P,i}^x\qubit{\#}\qubit{y}=\qubit{\#}\qubit{\psi_{x,y,i}}$ holds for a certain quantum state $\qubit{\psi_{x,y,i}}$.}
A trivial example of a committed prover is the
prover $P_{I}$, who always applies the identity operator. A
committed prover lets the verifier safely make a number of moves
without any {\lq\lq}direct{\rq\rq} interaction with him.
Observe that this new QIP model with committed provers
is in essence closer to
a circuit-based QIP model of Watrous \cite{Wat03}
than the original QIP model is. For convenience,
we name our new model an {\em
interaction-bounded QIP system} and use the new notation
$\qip^{\#}(1qfa)$ for the class of all languages recognized
with bounded error probability by such interaction-bounded QIP systems
with 1qfa verifiers.
Note that standard QIP systems can be naturally
transformed into interaction-bounded QIP systems by (possibly) modifying the blank symbol appropriately to a fresh non-blank symbol. This simple fact implies that
$\qip^{\#}(1qfa)$ contains $\qip(1qfa)$, which equals $\mathrm{REG}$ \cite{NY09}.

\begin{lemma}\label{qipregular}
$\mathrm{REG}\subseteq \qip^\#(1qfa)$.
\end{lemma}

We are now ready to clarify the meaning of the {\em number of interactions} in an interaction-bounded QIP system $(P,V)$.
Let us consider any non-halting global configuration in which
$V$ on input $x$ communicates with a prover (that is, writes a non-blank symbol in the communication cell).
For convenience, we call such a global configuration a {\em query configuration} and, at such a query configuration, $V$ is said to {\em query} a symbol to the prover. Recall from Section \ref{sec:basic-def} the definition of global computation paths.
The {\em number of interactions} in a given computation means the maximum number,
over all global computation paths $\chi$ of the computation of $(P,V)$,
of all
query configurations of non-zero amplitudes along $\chi$.
Let $L$ be any language and assume that $(P,V)$ recognizes $L$.
We say that the QIP protocol {\em $(P,V)$ makes $i$ interactions} on input $x$ if $i$ equals the number of interactions during the computation of $(P,V)$ on $x$. Furthermore, we call the QIP system $(P,V)$ {\em $k$-interaction bounded}\footnote{When $x\in L$ and $P^*$ is a malicious prover, Condition (i) does not impose any restriction on the number of interactions for $(P^*,V)$ on $x$. Instead of Conditions (i)--(ii), we could take a much stronger condition; for example, for every $x$ and every committed prover $P^*$, $(P^*,V)$ makes at most $k$ interactions. Such a stronger condition actually makes simpler the proof of, say, Proposition \ref{zero-1qfa}.} if (i) for every $x\in L$, the protocol $(P,V)$ makes at most $k$ interactions on the input $x$ and (ii) for every $x\notin L$ and for every committed prover $P^*$, the protocol $(P^*,V)$ makes at most $k$ interactions on the input $x$. At last, let $\qip^{\#}_{k}(1qfa)$ denote the class
of all languages recognized with bounded error probability by $k$-interaction bounded QIP systems with 1qfa verifiers. Since verifiers can control the number of queries, it is not difficult to show that $\mathrm{1QFA}\subseteq \qip^{\#}_{k}(1qfa)\subseteq \qip^{\#}_{k+1}(1qfa)\subseteq \qip^{\#}(1qfa)$ for any constant $k\in\nat$. In particular, $\qip^{\#}_{0}(1qfa) = \mathrm{1QFA}$ holds.

As the main theorem of this section, we want to
show in Theorem \ref{query-once} that (i) $1$-interaction helps a
verifier but (ii) $1$-interaction does not achieve the full power of $\qip^{\#}(1qfa)$.

\begin{theorem}\label{query-once}
$\qip^{\#}_0(1qfa) \subsetneqq\qip^{\#}_1(1qfa)
\subsetneqq\qip^{\#}(1qfa)$.
\end{theorem}

Theorem \ref{query-once} is a direct consequence of Lemma \ref{odd-1qfa} and Proposition \ref{zero-1qfa}, and its proof proceeds as follows.
For the first inequality of the theorem, we shall take the language $Odd$ defined as
the set of all binary strings of the form $0^m1z$, where $m\in\nat$,
$z\in\{0,1\}^*$, and $z$ contains an odd number of $0$s. Since $Odd\not\in \mathrm{1QFA}$ \cite{AKV01},
we shall show in Lemma \ref{odd-1qfa} that $Odd$ belongs to $\qip^{\#}_1(1qfa)$.
For the second inequality of the theorem, recall the regular language $Zero=\{x0\mid x\in\{0,1\}^*\}$ from Section \ref{sec:public-coin}.
We shall demonstrate in Proposition \ref{zero-1qfa} that  $\qip^{\#}_1(1qfa)$ does not include $Zero$. Since $\mathrm{REG}\subseteq\qip^{\#}(1qfa)$ by Lemma \ref{qipregular}, $Zero$ must belong to $\qip^{\#}(1qfa)$, and thus we shall obtain the desired separation of $\qip^{\#}(1qfa)$ from $\qip^{\#}_{1}(1qfa)$.
It therefore suffices to prove Lemma \ref{odd-1qfa} and Proposition \ref{zero-1qfa}.

As the first step, we prove Lemma \ref{odd-1qfa} that asserts $Odd\in\qip^{\#}_{1}(1qfa)$.

\begin{lemma}\label{odd-1qfa}
$Odd\in\qip^{\#}_{1}(1qfa)$.
\end{lemma}

\begin{proof}
We shall design a 1-interaction bounded QIP system $(P,V)$ that recognizes $Odd$.
Now, let $\Sigma=\{0,1\}$ and $\Gamma=\{\#,a\}$ be respectively an input alphabet
and a communication alphabet for $(P,V)$.
Moreover, let $Q=\{q_0,q_1,q_2,q_{acc},q_{rej,0},q_{rej,1}\}$ be a set of $V$'s inner states with $Q_{acc}=\{q_{acc}\}$ and $Q_{rej}=\{q_{rej,0},q_{rej,1}\}$. The protocol of the verifier $V$ is described as follows.
Table \ref{fig:interaction} gives
a formal description of $V$'s $(Q,\Gamma)$-transitions.
Making no query to a committed prover, $V$ continues to read input symbols until its tape head scans $1$ on the input tape.
When $V$ reads $1$, he queries the symbol $a$ to a committed prover by
executing the transition $V_1\qubit{q_0}\qubit{\#}=\qubit{q_0}\qubit{a}$.
If the prover returns $a$, then $V$ immediately rejects the input.
Otherwise, $V$ checks whether the substring of the input after $1$ includes an odd number of $0$s.
This check can be done by $V$ alone by applying $V_b\qubit{q_1}\qubit{\#}$ and $V_b\qubit{q_2}\qubit{\#}$ for $b\in\Sigma$.
The role of the honest prover $P$ is to work as an {\em eraser}, which erases any non-blank symbol written in the communication cell, to help the verifier safely make a transition from the inner state $q_0$ to $q_1$. Note that,  without the eraser, $V$ alone cannot make such a transition
because the unitary requirement of $V$'s strategy (namely, two orthonormal vectors are transformed into two orthonormal ones)
prohibits him from transforming both $|q_1\rangle|a\rangle$
and $|q_1\rangle|\#\rangle$ into $|q_1\rangle|\#\rangle$.
To be more precise, whenever receiving the symbol $a$ from $V$, $P$ returns the symbol $\#$ and copies $a$ into the first blank cell of his private tape. Technically speaking, to make $P$ unitary, we need to map other visible configurations $\qubit{\#}\qubit{y}$ for certain $y$'s not having appeared in $P$'s private tape to superpositions of the form  $\qubit{a}\qubit{\phi_{x,y}}$ with an appropriate vector $\qubit{\phi_{x,y}}$. By right implementation, it is possible  to make $P$ a committed prover.

\begin{table}[ht]
\bs\begin{center}
\begin{tabular}{|lll|}\hline
$V_{\cent}\qubit{q_0}\qubit{\#}=\qubit{q_0}\qubit{\#}$
& $V_{0}\qubit{q_0}\qubit{\#}=\qubit{q_0}\qubit{\#}$ &
$V_{1}\qubit{q_0}\qubit{\#}=\qubit{q_1}\qubit{a}$ \\
$V_{\$}\qubit{q_0}\qubit{\#} =\qubit{q_{rej,0}}\qubit{\#}$ &
$V_{0}\qubit{q_1}\qubit{\#}=\qubit{q_2}\qubit{\#}$ &
$V_{1}\qubit{q_1}\qubit{\#}=\qubit{q_1}\qubit{\#}$ \\
$V_{\$}\qubit{q_1}\qubit{\#} =\qubit{q_{rej,1}}\qubit{\#}$ &
$V_{0}\qubit{q_2}\qubit{\#} =\qubit{q_1}\qubit{\#}$ &
$V_{1}\qubit{q_2}\qubit{\#} =\qubit{q_2}\qubit{\#}$ \\
$V_{\$}\qubit{q_2}\qubit{\#} =\qubit{q_{acc}}\qubit{\#}$
& $V_{0}\qubit{q_1}\qubit{a}=\qubit{q_{rej,0}}\qubit{\#}$ &
$V_{1}\qubit{q_1}\qubit{a}=\qubit{q_{rej,0}}\qubit{\#}$ \\  \hline
\end{tabular}
\caption{{\small $(Q\times\Gamma)$-transitions $\{V_{\sigma}\}_{\sigma\in\check{\Sigma}}$ of $V$ for $Odd$}}\label{fig:interaction}
\end{center}
\end{table}

Next, we shall show that $(P,V)$ recognizes $Odd$ with probability $1$. Let $x$ be any binary input.
First, consider the case where $x$ is in $Odd$. Assume that $x$ is of the form $0^m1y$,
where $y$ contains an odd number of 0s.
The honest prover $P$ erases $a$ that is sent from $V$ when $V$ reads $1$.
This process helps $V$ shift to the next mode of checking the number of $0$s.
Since $V$ can check whether $y$ includes an odd number of $0$s without any communication with the prover, $V$ eventually accepts $x$ with certainty.
On the contrary, assume that $x\not\in Odd$. In the special case where $x\in\{0\}^*$, $V$ can reject $x$ with certainty with no query to a committed prover.
Now, focus our attention to the remaining case where $x$ contains a $1$.
Let us assume that $x$ is of the form $0^m1y$, where $y$ contains an even number of 0s. The verifier $V$ sends $a$ to a committed prover when he reads $1$.
Note that $V$'s protocol is essentially deterministic. To maximize the acceptance probability of $V$, a malicious prover needs to return $\#$ to $V$ since, otherwise, $V$ immediately rejects $x$ in a deterministic fashion.
Since $V$ can check whether $y$ includes an odd number of $0$s
without making any query to the prover, for any committed prover $P^*$, the QIP protocol $(P^*,V)$ rejects $x$ with certainty.
Since the number of interactions made by the protocol is obviously
at most $1$,  $Odd$ therefore belongs to $\qip_1^\#(1qfa)$, as requested.
\end{proof}

As the second step, we need to prove Proposition \ref{zero-1qfa} regarding the language $Zero=\{x0\mid x\in\{0,1\}^*\}$. This regular language $Zero$ is known to be outside of $\mathrm{1QFA}$ \cite{KW97};
in other words, $Zero\not\in \qip^{\#}_{0}(1qfa)$ since $\qip_0^{\#}(1qfa) = \mathrm{1QFA}$.
Proposition \ref{zero-1qfa} expands this impossibility result and shows that $Zero$ is not even in $\qip^{\#}_{1}(1qfa)$.

\begin{proposition}\label{zero-1qfa}
$Zero\not\in\qip^{\#}_{1}(1qfa)$.
\end{proposition}

Since the proof of Proposition \ref{zero-1qfa} is quite
involved, it will be given in the subsequent subsection.

\subsection{Proof of Proposition \ref{zero-1qfa}}

Our proof of Proposition \ref{zero-1qfa} proceeds by way of contradiction.  Towards a contradiction, we
start with assuming $Zero\in \qip_{1}^{\#}(1qfa)$ and take
a 1-interaction bounded QIP system $(P,V)$ with 1qfa verifier $V$ that recognizes $Zero$   with error probability at most $1/2-\eta$
for a certain constant $\eta>0$.  Our goal is to pick a suitable string $\tilde{y}0$ and an appropriate prover $P'_y$ and to prove that their
associated  protocol $(P'_y,V)$ accepts $\tilde{y}01^m$ with probability at least $1/2$, because this contradicts our assumption that $Zero\in \qip_{1}^{\#}(1qfa)$. For this purpose, we shall employ a technical tool, called {\lq\lq}query weight{\rq\rq}, which is the sum of all the squared magnitudes of query configurations appearing in a computation of the protocol $(P',V)$ on a specified input.
However, a different choice of provers may result in different query weights,
as seen in the second and the third computation trees shown in Figure \ref{fig:query}. To cope with this unfavorable situation, we shall introduce another computation model, which is not dependent on the choice of provers, and we shall prove that this model gives an upper-bound of the query weight induced by any committed prover.
Using this model and its query weight, we shall finally select the desired string $\tilde{y}0$ and the desired prover $P'_y$.

Now, let $Q$ be a set of $V$'s inner states
and let $\Gamma$ be a communication alphabet. Write $\Sigma$ for our input alphabet $\{0,1\}$. For technicality, we assume without loss of generality  that $V$ never queries at the very time when it enters a certain halting inner state, that is, the time when $V$'s tape head scans $\$$.

First, we introduce two useful notions: {\lq\lq}1-interaction condition{\rq\rq} and {\lq\lq}query weight.{\rq\rq}
Fix an input $x\in\Sigma^*$ and let $P'$ be any committed prover. For readability, we use the notation $Comp_{V}(P',x)$ to denote a
computation of the QIP protocol $(P',V)$ on input $x$ when $P'$ takes
strategy $P'_x$.

A committed prover $P'$ is said to meet the {\em 1-interaction condition at $x$ with $V$} if the corresponding protocol $(P',V)$ makes at most $1$ interaction on input $x$.
Note that, when $P'$ satisfies the 1-interaction condition at $x$ with $V$, for any query configuration  $\xi$ of non-zero amplitude along computation path $\chi$ in $Comp_{V}(P',x)$, there exists no other query configuration in $Comp_{V}(P',x)$ between the initial configuration and this given  configuration $\xi$ along the computation path $\chi$.
Let $C^{(1)}_{x,V}$ be the collection of all committed provers $P'$
who satisfy the 1-interaction condition at $x$ with $V$.
It is important to note that, whenever a prover in $C^{(1)}_{x,V}$ answers to $V$ with non-blank communication symbols
with non-zero amplitude, $V$ must change these symbols back to the blank symbol immediately since, otherwise, $V$ is considered to make the second query
in the next round, according to the
definition of our interaction-bounded QIP model.

Now, let us choose any prover $P'$ in $C^{(1)}_{x,V}$ and consider its  computation $Comp_{V}(P',x)$. By introducing an extra projection operator,
we modify $Comp_{V}(P',x)$ as follows.
Whenever $V$ performs a projective measurement onto $(W_{non},W_{acc},W_{rej})$, we then apply to the communication cell an additional projection that maps onto the Hilbert space spanned by $\qubit{\#}$. This projection makes all non-blank symbols collapse.
The protocol $(P,V)$ then continues to the next step.
By Condition (*) in Section \ref{sec:new-QIP-model}, observe that the computation obtained by inserting an extra projection operator at every step of $V$ is independent of the choice of  committed provers. To express this modified computation of $V$ on $x$,
we use another notation $MComp_{V}(x)$.
Figure \ref{fig:query} illustrates the difference between such a modified computation and two original computations generated by two different provers $P_1$ and $P_2$.

For two strings $x,y\in\Sigma^*$, the {\em query weight $wt_{V}^{(x)}(y)$ of $V$ at $y$ conditional to $x$} is the sum of all the squared magnitudes of
the amplitudes of query configurations appearing in $MComp_{V}(xy)$ while $V$ is reading $y$.
For brevity, let $wt_{V}(y)=wt_{V}^{(\lambda)}(y)$, where $\lambda$ is the empty string. Note that this query weight $wt_{V}^{(x)}(y)$
ranges between $0$ and $1$ and satisfies that $wt_{V}(x) + wt^{(x)}_{V}(y) = wt_{V}(xy)$ for any strings $x,y\in\Sigma^*$.

\begin{figure}[t]
\begin{center}
\centerline{\psfig{figure=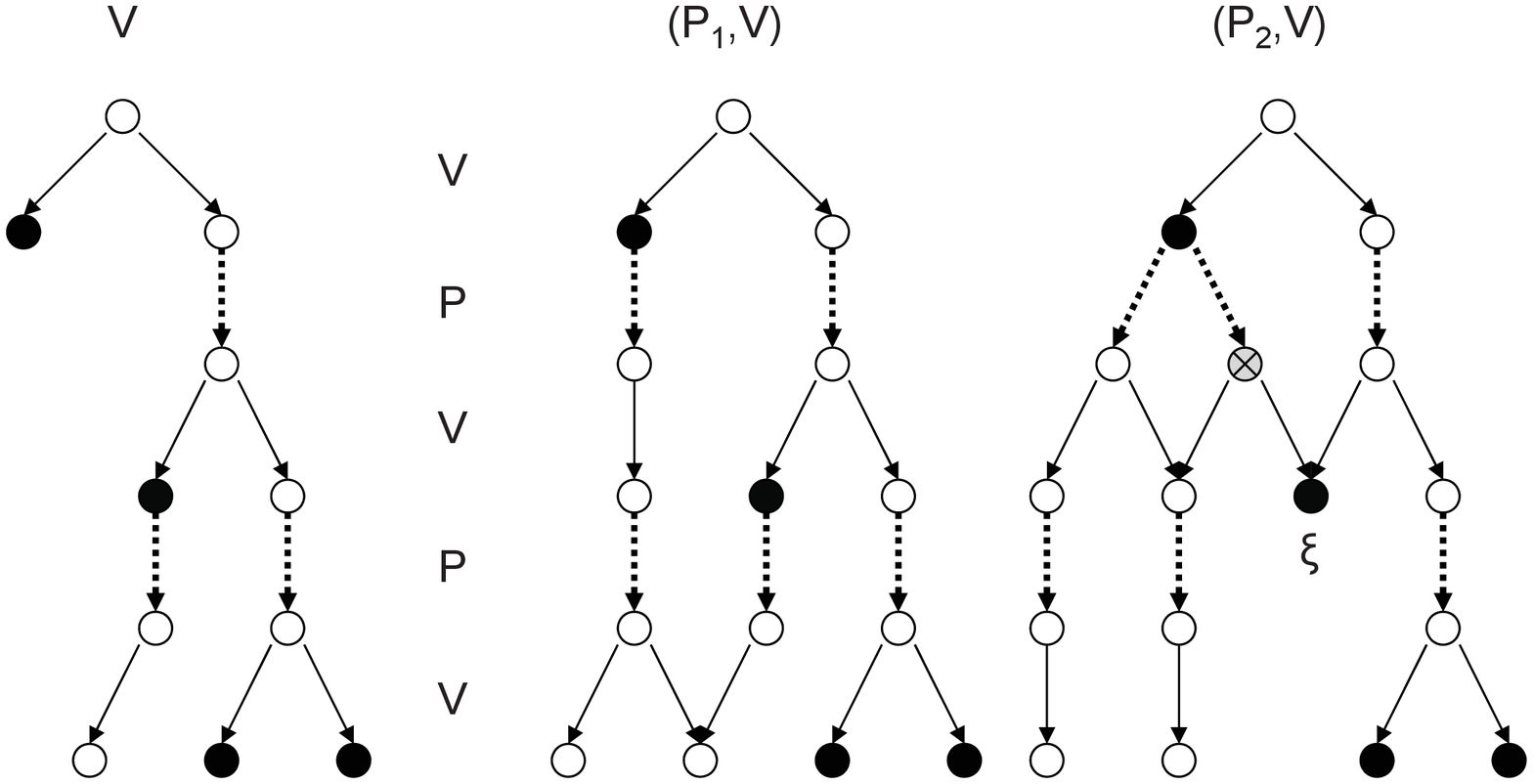,height=4.5cm}}
\caption{{\small Examples of a modified computation. The leftmost graph depicts a  modified computation of $V$ on input $x$. Two remaining graphs are original computations generated by $V$ on $x$ using different provers $P_1$ and $P_2$. The black circles indicate query configurations whereas  the white circles indicate non-query configurations. Let query configuration $\xi$ marked black have
zero amplitude.  The crossed grey circle (located at the upper-left side of $\xi$) is the place where prover $P_2$ forces $V$ to generate a new computation path
that destructively interferes with an existing path in the modified computation of $V$.}}\label{fig:query}
\end{center}
\end{figure}

Recall from our assumption that $(P,V)$ is a 1-interaction bounded QIP system  recognizing $Zero$ with success probability at least $1/2+\eta$.
The following lemma shows two properties of  the query weights of $V$.

\begin{lemma}\label{weight}
Let $P'$ be any committed prover and let $x,y$ be any strings.
\begin{enumerate}\vs{-2}
  \setlength{\topsep}{0mm}%
  \setlength{\itemsep}{1mm}
  \setlength{\parskip}{0cm}%
  
\item If $P'\in C^{(1)}_{x,V}$, then, for every query configuration $\xi$ of non-zero amplitude in $Comp_{V}(P',x)$, any computation path $\chi$ in $Comp_{V}(P',x)$ ending with $\xi$ appears in $MComp_{V}(x)$ ending with $\xi$ of the same amplitude.

\item If $P'\in C^{(1)}_{xy,V}$, then $wt_{V}^{(x)}(y)$ is greater than or equal to the sum of all the squared magnitudes
of amplitudes of query configurations in $Comp_{V}(P',xy)$ while $V$'s tape head is reading $y$.
\end{enumerate}
\end{lemma}

\begin{proof}
(1) Take any committed prover $P'$ in $C^{(1)}_{x,V}$. Let $\xi$ be any query configuration in $Comp_{V}(P',x)$ having non-zero amplitude, say, $\alpha_{\xi}$.
Since $\alpha_{\xi}$ is not zero in $Comp_{V}(P',x)$, there must exist at least one computation path in $Comp_{V}(P',x)$ ending with $\xi$ of non-zero amplitude. Let us consider such a computation path, say, $\chi$. Since $P'$ satisfies the 1-interaction condition, $\chi$ cannot contain any query configuration of non-zero amplitude except for the last configuration $\xi$. By the definition of $MComp_{V}(x)$, no projective measurement on the communication cell is performed along $\chi$. Hence, all configurations inside $\chi$ must be present also in $MComp_{V}(x)$. Thus, $\chi$ appears in $MComp_{V}(x)$. Since $\chi$ is arbitrary, all the computation paths  in $Comp_{V}(P',x)$ ending with $\xi$,  which contribute to the amplitude $\alpha_{\xi}$, must appear in $MComp_{V}(x)$. Therefore, the amplitude of $\xi$ in $MComp_{V}(x)$ equals $\alpha_{\xi}$, as requested.

(2) Assume that $P'\in C^{(1)}_{xy,V}$.
Let us recall that the query weight $wt_{V}^{(x)}(y)$ is the sum of all the squared magnitudes of the amplitudes of query configurations in $MComp_{V}(xy)$ while $V$'s reading $y$.
By (1), for every query configuration $\xi$ of non-zero amplitude in $Comp_{V}(P',xy)$, the squared magnitude of the amplitude of $\xi$ in $Comp_{V}(P',xy)$ is equal to that of $\xi$ in $MComp_{V}(xy)$.
Note that the converse in general may not be true; that is, there may be a query configuration of non-zero amplitude in $MComp_{V}(x)$ that never appears in $Comp_{V}(P',x)$. Figure \ref{fig:query} illustrates such a case.  By summing up the squared magnitudes over all query configurations $\xi$ in $Comp_{V}(P',xy)$, we immediately obtain (2).
\end{proof}

We continue the proof of Proposition \ref{zero-1qfa}. Let us consider a value $\nu$ that is the supremum, over all strings $w$ in $Zero$, of the query weight of $V$ at $w$; namely, $\nu = \sup_{w\in Zero}\{wt_{V}(w)\}$. Observe that $0\leq \nu\leq 1$ since any query weight is in the unit real interval $[0,1]$.
For readability, we omit the letter $V$ whenever it is clear from the context. Let $P_I$ denote a committed prover applying only the identity operator at every step.

\begin{claim}\label{nu-is-positive}
$\nu>0$.
\end{claim}

\begin{proof}
Let us assume that $\nu=0$. From this assumption, it follows that
$wt(x)=0$ for all $x\in Zero$.
To obtain a contradiction, we aim at constructing
an appropriate bounded-error 1qfa $M$ that recognizes $Zero$.
Recall that $P$ is a honest committed prover that makes the 1-interaction bounded QIP system $(P,V)$ recognize $Zero$.
Firstly, we assert that even a simple protocol $(P_I,V)$ can recognize $Zero$ with success probability at least $1/2+\eta$. For each input $x\in Zero$, since $wt(x)=0$, Lemma \ref{weight}(2) implies that all query configurations in $Comp_{V}(P,x)$ must have zero amplitudes. This situation implies that,   in the superposition of global configurations at each step of the computation of $(P,V)$ on $x$,  the verifier $V$'s next moves cannot be affected by any messages sent out by $P$. Hence, we can replace $P$ by $P_I$  without changing the outcome of $V$.

Next, we shall construct the desired 1qfa $M$. Any inner state of $M$ has the form $(q,\sigma)$, which implicitly reflects both $V$'s inner state $q$ and a symbol $\sigma$
in the communication cell.
Our 1qfa $M$ behaves as follows.
On input $x$, $M$ simulates $V$ on $x$ using the {\lq\lq}imaginary{\rq\rq} prover $P_I$ by maintaining the content of the communication cell
as an integrated part of $M$'s inner states.

Finally, we claim that $M$ recognizes $Zero$ with bounded error probability.
Let $x$ be any input string.
If $x$ is in $Zero$, then, since any query configuration with
the prover $P_{I}$ has the zero amplitude,
$M$ correctly accepts $x$ with probability $\geq 1/2+\eta$.
Likewise, if $x$ is not in $Zero$, then the protocol $(P_{I},V)$ rejects $x$ with probability $\geq 1/2+\eta$; thus, $M$  also rejects $x$ with the same probability.
Therefore, $M$ recognizes $Zero$ with error probability at most $1/2-\eta$, as requested. Since $Zero\notin \mathrm{1QFA}$, we obtain a contradiction, and therefore $\nu>0$ follows.
\end{proof}

Next, we shall construct
a committed prover $P'$ and a string $z\notin Zero$ that force the protocol $(P',V)$ to  accept $z$ with probability at least $1/2$.
Let us recall from Section \ref{sec:basic-def} the notation $P_{w}$, which
refers to a strategy $\{U_{P,i}^{w}\}_{i\in\nat^{+}}$ of $P$ on input $w$. Since $\nu>0$ by Claim \ref{nu-is-positive},
for every real number $\gamma\in(0,\nu]$, there exists a string $w$ in $Zero$ such that $wt(w) \geq \nu-\gamma$.
Given any $y\in\Sigma^*$, we set $\gamma_y=\min\{\eta^2/16(|y|+1)^2,\nu\}$
and choose the lexicographically minimal string $w_{y}\in Zero$ satisfying  $wt(w_y) \geq \nu - \gamma_{y}$. For readability, we abbreviate the string $w_{y}y$ as $\tilde{y}$.

Moreover, we define a new committed prover $P'_{y}$ that behaves on input $\tilde{y}01^m$ ($=w_{y}y01^m$),
where $m\in\nat$,
in the following fashion:
$P'_{y}$ follows the strategy $P_{\tilde{y}0}$
while $V$'s tape head is reading $\cent w_{y}$ and then $P'_{y}$ behaves as $P_{I}$  while $V$ is reading the remaining portion $y01^m\$$.   Since $P$ satisfies the 1-iteration condition, we obtain $P_{\tilde{y}0}\in C^{(1)}_{\tilde{y}0,V}$. By its definition, $P'_y$ also belongs to $C^{(1)}_{\tilde{y}0,V}$.
As for the notation $p_{acc}(x,P,V)$ introduced in Section \ref{sec:basic-def}, we simply drop ``$V$'' and write $p_{acc}(x,P)$ instead. We then claim the following lower bound of $p_{acc}(\tilde{y}0,P'_y)$.

\begin{claim}\label{low-bound-Pacc}
For any string $y\in\Sigma^*$, $p_{acc}(\tilde{y}0,P'_{y}) \geq 1/2+\eta/2$.
\end{claim}

\begin{proof}
Let $y$ be an arbitrary input string. From the fact $\tilde{y}0\in Zero$, it must hold that $p_{acc}(\tilde{y}0,P_{\tilde{y}0})\geq 1/2+\eta$ because  $(P,V)$ recognizes $Zero$.
Note that, on the same
input $\tilde{y}0$, the protocol $(P'_{y},V)$ works in the same way as $(P_{\tilde{y}0},V)$
while $V$ is reading $\cent w_{y}$.
Consider the query weight $wt(\tilde{y}0)$. Since $wt(\tilde{y}0)\leq \nu$,
we obtain $wt(w_y) + wt^{(w_y)}(y0) = wt(\tilde{y}0) \leq \nu$, from which
it follows that $wt^{(w_y)}(y0)\leq \gamma_{y}$ using the inequality  $wt(w_y)\geq \nu - \gamma_{y}$.
Lemma \ref{weight}(2) implies that, for any committed prover $P^*$ in $C^{(1)}_{\tilde{y}0,V}$,
$wt^{(w_y)}(y0)$ upper-bounds the sum of all the squared magnitudes of the amplitudes of query configurations in $Comp_{V}(P^*,\tilde{y}0)$
while the tape head is reading $y0\dollar$.
Notice that both $P_{\tilde{y}0}$ and $P'_y$ belong to  $C^{(1)}_{\tilde{y}0,V}$.
Taking $P_{\tilde{y}0}$ and $P'_{y}$ as $P^*$, a simple calculation shows  that
\begin{eqnarray*}
\left|p_{acc}(\tilde{y}0,P'_{y}) - p_{acc}(\tilde{y}0,P_{\tilde{y}0})\right|
&\leq& 2\left(wt^{(w_y)}(y0)\right)^{1/2}|y0|
\;\;\leq\;\; 2\sqrt{\gamma_{y}}(|y|+1) \\
&\leq& 2\sqrt{\frac{\eta^2}{16(|y|+1)^2}}\cdot(|y|+1)
\;\;=\;\; \eta/2,
\end{eqnarray*}
where the first inequality is shown as in, \eg  \cite[Lemma 9]{Yam03}. Since $p_{acc}(\tilde{y}0,P_{\tilde{y}0})\geq 1/2+\eta$,
it follows that $p_{acc}(\tilde{y}0,P'_{y}) \geq p_{acc}(\tilde{y}0,P_{\tilde{y}0}) - \eta/2
\geq (1/2+\eta)- \eta/2 \geq 1/2+\eta/2$.
\end{proof}

Henceforth, we shall turn  our attention to the protocol $(P'_y,V)$ working on input $\tilde{y}01^m$ for any number $m\in\nat^{+}$ and we shall estimate its acceptance probability $p_{acc}(\tilde{y}01^m,P'_y)$.

The initial superposition of global configurations is $\qubit{q_0}\qubit{\#}\qubit{\#^{\infty}}$, where we omit the qubits representing the tape head position of $V$, because the tape head moves only in one direction without stopping.
Let $\VV=\mathrm{span}\{\qubit{q}\mid q\in Q\}$, let $\MM=\mathrm{span}\{\qubit{\sigma}\mid \sigma\in \Gamma\}$, and let $\PP$ be a Hilbert space representing the prover's private tape.
For each number $m\in\nat^{+}$,
we denote by $\qubit{\psi_{y,m}}$ a superposition in the global configuration space $\VV\otimes\MM\otimes\PP$ obtained by $(P'_{y},V)$ on input $\tilde{y}01^m$
just after $V$'s tape head moves off the right end of $\cent \tilde{y}0$ and $P'_y$ then applies the identity operator while reading the
remaining $1^m\$$. Since $P'_{y}$ basically does nothing after $V$ has read $\cent w_{y}$, $\qubit{\psi_{y,m}}$ does not depend on the choice of $m$. We therefore write $\qubit{\psi_y}$ instead of $\qubit{\psi_{y,m}}$, for simplicity. Moreover, we set $\mu = \mathrm{inf}_{y\in\Sigma^*}\{\|\qubit{\psi_y}\|\}$.

\begin{claim}\label{mu_is_zero}
$\mu>0$.
\end{claim}

\begin{proof}
Let us prove this claim by contradiction. First, assume that $\mu=0$.
Since this assumption means $\mathrm{inf}_{y\in\Sigma^*}\{\|\qubit{\psi_y}\|^2\}=0$
and the constant $\eta$ is positive by its definition,
there exists a string $y\in\Sigma^*$ satisfying  $\|\qubit{\psi_{y}}\|^2<\eta$.
Now, let us consider a special input string $\tilde{y}01$.
Since $\tilde{y}01\not\in Zero$, it must follow that $p_{acc}(\tilde{y}01,P'_y)\leq 1/2-\eta$
by the soundness condition of the 1-interaction bounded QIP system $(P,V)$.
Notice that $\qubit{\psi_y}$ is composed only of non-halting configurations having non-zero amplitudes. For convenience, let $\alpha_y$ denote the total acceptance probability of $V$ obtained while $V$'s reading $\tilde{y}0$ (not including $\dollar$). It holds that $p_{acc}(\tilde{y}0,P'_y)\leq \alpha_y + \|\qubit{\psi_y}\|^2$. While $V$ is reading $\tilde{y}0$ out of $\tilde{y}01$, $P'_y$ takes exactly the same strategy as $P$ does on the input $\tilde{y}0$. We thus conclude that $p_{acc}(\tilde{y}01,P'_y)$ is at least $\alpha_y$. By combining those inequalities,
we obtain  $p_{acc}(\tilde{y}0,P'_y) \leq p_{acc}(\tilde{y}01,P'_y)+\|\qubit{\psi_y}\|^2$.
Since $p_{acc}(\tilde{y}0,P'_y)\geq 1/2+\eta/2$ by
Claim \ref{low-bound-Pacc},
it follows that $\|\qubit{\psi_y}\|^2 \geq p_{acc}(\tilde{y}0,P'_y) - p_{acc}(\tilde{y}01,P'_y)\geq 1/2+\eta/2 - (1/2-\eta) >\eta$,
a contradiction.
Therefore, $\mu>0$ must hold.
\end{proof}

Let $\epsilon$ be any sufficiently small positive real number and choose  the lexicographically first string $y$ for which  $\|\qubit{\psi_y}\|\in[\mu,\mu+\epsilon)$.
Let $I_1$ and $I_2$ denote the identity operators acting on $\PP$ and  $\VV$, respectively, and write $W$ for
$(I_2\otimes P_I)(E_{non}\otimes I_1)(V_1\otimes I_1)$, where $V_{\sigma}$ is a $(Q\times \Gamma)$-transition of $V$ on symbol $\sigma$.
For any integer $j\geq 1$, it holds that
$\mu\leq \|W^j \qubit{\psi_y}\| <\mu+\epsilon$ by the definition of $\mu$.  Note that, since $y$ is fixed, the prover $P'_y$ uses only a finite number of tape cells on his private tape. Let $d$ denote the maximal number of cells in use by $P'_y$. Without loss of generality, we can assume that the Hilbert space $\PP$ is of dimension $d$.
Thus, all vectors in a series $\{W^{j}\qubit{\psi_y}\mid j\in\nat^{+}\}$ are  in a Hilbert space of dimension $d|\Sigma||\Gamma|$. Since $\mu>0$ by Claim \ref{mu_is_zero},
there are two numbers $j_1,j_2\in\nat$ with $j_i<j_2$ such that $\|W^{j_1}\qubit{\psi_y} - W^{j_2}\qubit{\psi_y}\|<\epsilon$.
By an analysis similar to \cite{KW97} (see also \cite[Lemma 4.1.10]{Gru99}),
there exist a constant $c>0$, which is independent of the value of $\epsilon$, satisfying that
$\|W^{j_1}\qubit{\psi_y} - W^{j_2}\qubit{\psi_y}\| \leq \|\qubit{\psi_y}-W^{m}\qubit{\psi_y}\|<c\cdot \epsilon^{1/4}$, where  $m=j_2-j_1\geq1$.
{}From this inequality follows
\[
\left|p_{acc}(\tilde{y}0,P'_{y}) - p_{acc}(\tilde{y}01^m,P'_{y})\right|
\leq \left\| V_{\$}\qubit{\psi_y} - V_{\$}W^m\qubit{\psi_y}\right\| = \left\| \qubit{\psi_y} - W^m\qubit{\psi_y} \right\| < c\epsilon^{1/4},
\]
where the first inequality is a folklore (see \cite[Lemma 8]{Yam03} for its proof).
Since $\epsilon$ is arbitrary, we can set $\epsilon =\left(\eta/2c\right)^4$.
Because $p_{acc}(\tilde{y}0,P'_{y})\geq 1/2+\eta/2$ by Claim \ref{low-bound-Pacc},
it follows that $p_{acc}(\tilde{y}01^m,P'_{y})\geq (1/2+\eta/2)-c\epsilon^{1/4} = 1/2$.
This contradicts our assumption that,  for any committed prover $P^*$,  $(P^*,V)$ accepts $\tilde{y}01^m$ with probability $\leq 1/2-\eta/2 < 1/2$.
Therefore, $Zero$ does not belong to $\qip^{\#}_{1}(1qfa)$, as requested.

We have finally completed the proof of Proposition \ref{zero-1qfa}.

\section{Challenging Open Questions}

Throughout Sections \ref{sec:classical-prover}--\ref{sec:interaction}, we have placed various restrictions on the behaviors of verifiers and provers in our qfa-verifier QIP systems and we have studied how those restrictions affect the language recognition power of the systems. The restricted models that we have considered in Sections \ref{sec:classical-prover}--\ref{sec:interaction} include: classical-prover QIPs, public QIPs, and interaction-bounded QIPs.   After an initial study of this paper, nonetheless, there still remain numerous unsolved questions concerning those QIP systems. Hereafter, we shall give a short list of the important open questions as a helping guide to future research.

(1) The relationships between quantum provers and classical provers are
still
not entirely clear in the context of qfa-verifier QIP systems, mostly because of the soundness condition imposed on the systems.
For example, we expect to see a fundamental separation between $\qip(1qfa)$ and $\qip(1qfa,c\mbox{-}prover)$ as well as between $\qip(2qfa,poly\mbox{-}time)$ and $\qip(2qfa,poly\mbox{-}time,c\mbox{-}prover)$.

(2) In general, we need to discover precise relationships between public-coin  IP systems (namely, AM systems) and public QIP systems beyond Theorem \ref{public-vs-am}. Moreover, associated with 2qfa verifiers,
we ask whether  $\mathrm{2QFA}(poly\mbox{-}time)$ is properly included in $\qip(2qfa,public)$ and whether  $\qip(2qfa,public,poly\mbox{-}time)$ is different from $\qip(2qfa,poly\mbox{-}time)$.  A separation between $\qip(2qfa)$ and  $\qip(2qfa,public)$ is also unknown and this
must be settled.

(3) A new interaction-bounded QIP system is of special interest in analyzing the roles of interactions between provers and verifiers. For this model,
we hope to see that the equality $\qip^{\#}(1qfa)=\qip(1qfa)$ indeed holds.
Unsolved so far is a general question of whether $k+1$ interactions are more powerful than $k$ interactions.
Since a 1qfa verifier is unable to count the number of interactions (or queries),
we may not directly generalize the proof of Theorem \ref{query-once}
to assert
that $\qip^{\#}_k(1qfa)\neq \qip^{\#}_{k+1}(1qfa)$ for any constant $k$ in $\nat^{+}$.
Nevertheless, we still conjecture that this assertion is true.

(4) It is of great interest to seek an algebraic characterization of our qfa-verifier QIP systems. Such a characterization may shed new light on a nature of quantum interactions between two parties.

\let\oldbibliography\thebibliography
\renewcommand{\thebibliography}[1]{%
  \oldbibliography{#1}%
  \setlength{\itemsep}{0pt}%
}
\bibliographystyle{alpha}

\end{document}